\documentclass[11pt]{article}
\usepackage[cp1251]{inputenc}
\usepackage[T2A]{fontenc}
\usepackage[english]{babel}
\usepackage{amsthm, amsmath, amssymb, amsbsy, amscd, amsfonts, latexsym, euscript,epsf}
\newtheorem{thm}{Theorem}[section]
\usepackage{bbold, bm}
\usepackage{epic,eepic}
\usepackage{texdraw}
\usepackage[dvips]{color}
\usepackage{color}
\usepackage[dvipsnames]{xcolor}
\usepackage{graphicx}
\usepackage{exscale,relsize}
\usepackage{authblk}
\usepackage{url}
\usepackage{enumerate}
\allowdisplaybreaks
\newtheorem{proposition}[thm]{Proposition}
\newtheorem{corollary}[thm]{Corollary}

\newtheorem{remark}[thm]{Remark}
\newtheorem{theorem}[thm]{Theorem}

\title{\textbf{Birational solutions to the set-theoretical 4-simplex equation}} 

\author{S. Konstantinou-Rizos\thanks{skonstantin84@gmail.com}}

\affil{Centre of Integrable Systems, P.G. Demidov Yaroslavl State University, Yaroslavl, Russia}

\theoremstyle{definition}
\newtheorem{definition}[thm]{Definition}

\DeclareMathOperator{\End}{End}

\DeclareMathOperator{\id}{id}

\begin{document}

\maketitle

\begin{abstract}
The 4-simplex equation is a higher-dimensional analogue of Zamolodchikov's tetrahedron equation and the Yang--Baxter equation which are two of the most fundamental equations of mathematical physics. In this paper, we introduce a method for constructing 4-simplex maps, namely solutions to the set-theoretical 4-simplex equation, using Lax matrix refactorisation problems. Employing this method, we construct 4-simplex maps which at a certain limit give tetrahedron maps classified by Kashaev, Korepanov and Sergeev. Moreover, we construct a Kadomtsev--Petviashvili type of 4-simplex map. Finally, we introduce a method for constructing 4-simplex maps which can be restricted on level sets to parametric 4-simplex maps using Darboux transformations of integrable PDEs. We construct a nonlinear Schr\"odinger type parametric 4-simplex map which is the first parametric 4-simplex map in the literature. 
\end{abstract}

\bigskip

\hspace{.2cm} \textbf{PACS numbers:} 02.30.Ik, 02.90.+p, 03.65.Fd.

\hspace{.2cm} \textbf{Mathematics Subject Classification 2020:} 35Q55, 16T25.

\hspace{.2cm} \textbf{Keywords:} 4-simplex equation, parametric 4-simplex maps, Zamolodchikov's tetrahedron

\hspace{2.4cm}  equation, parametric tetrahedron maps, Darboux transformations, NLS type

\hspace{2.4cm}  equations, Yang--Baxter maps.

\section{Introduction}
The $n$-simplex equation \cite{Bazhanov} is a generalisation of the tetrahedron equation ($n=3$) --- firstly introduced by Zamolodchikov \cite{Zamolodchikov, Zamolodchikov-2} --- and the well-celebrated Yang--Baxter equation ($n=2$), which are two of the most fundamental equations of mathematical physics. The study of the set-theoretical $n$-simplex equation was formally initiated by Drinfeld \cite{Drinfeld} for $n=2$. We use the term $n$-simplex map when referring to the solutions of such equations as it was suggested by Buchstaber \cite{Buchstaber}, as well as Veselov \cite{Veselov}. It is known that $n$-simplex maps play an important role in the theory of integrable systems of mathematical physics. In fact, they are associated with integrable lattice equations via symmetries (see, e.g., \cite{Pap-Tongas, Kassotakis-Tetrahedron} for 2- and 3-simplex maps) and integrable nonlinear PDEs, as well as they are related to Darboux and B\"acklund transformations \cite{Sokor-Sasha, Sokor-Papamikos, Papamikos},  inverse scattering problems  and  other integrable objects. Therefore, there is a need to develop methods for constructing interesting $n$-simplex maps which may give rise to important integrable models.

Moreover, $n$-simplex maps are related to polygon equations \cite{Dimakis-Hoissen, Kashaev-Sergeev-2}, and also can be generated by matrix refactorisation problems. In particular, a quite interesting fact is that the local $(n-1)$-simplex equation is a generator of solutions to the $n$-simplex equation \cite{Maillet, Maillet-Nijhoff, Dimakis-Hoissen}. Following \cite{Veselov-Suris}, if an $n$-simplex map is generated by the local $(n-1)$-simplex equation, then the latter is called its Lax representation \cite{Dimakis-Hoissen}. However, in the literature there are only examples of 2- and 3-simplex maps generated by the local 1- and 2-simplex equations, respectively. 

This paper is concerned with the study of the 4-simplex equation, the so-called Bazhanov--Stroganov equation, and the development of new methods for constructing set-theoretical solutions to it. In principle, the 4-simplex equation has more solutions than the tetrahedron and the Yang--Baxter equation, since all the Yang--Baxter and tetrahedron maps can be extended trivially to 4-simplex maps. However, only a small number of 4-simplex solutions are known in the literature up-to-date; some simple solutions can be found in \cite{Hietarinta, Bardakov}.

In this paper, we show how to obtain new solutions to the set-theoretical 4-simplex equation by generalising the matrix generators of tetrahedron maps. We apply this method to the list of generators of the Kashaev--Sergeev--Korepanov tetrahedron maps \cite{Kashaev-Sergeev} and a KP type tetrahedron map, and we construct novel 4-simplex maps. Moreover, we extend the ideas of \cite{Sokor-Sasha, Sokor-2020}; namely, we employ Darboux transformations and show how to systematically construct 4-simplex maps which can be restricted to parametric 4-simplex maps on invariant leaves. As an illustrative example, we present a new rational parametric 4-simplex map of Nonlinear Schr\"odinger (NLS) type.


\subsection{Organisation of the paper}

This paper is organised as follows:

In the next section, we introduce the notation we use throughout the text and we give the definitions of Zamolodchikov's tetrahedron equation and the 4-simplex equation. Moreover, we explain the relation between the former equations and particular matrix refactorisation problems. Specifically, we clarify the relation between the tetrahedron equation and the local Yang--Baxter equation, as well as the relation between the 4-simplex equation and the local tetrahedron equation.

In section \ref{Extension method}, we present a method for extending tetrahedron maps to solutions of the 4-simplex equation by generalising their matrix generators. The obtained 4-simplex maps can be restricted to tetrahedron maps at a particular limit. We employ this method to construct new 4-simplex maps of Kashaev--Korepanov--Sergeev type. Moreover, we construct a new 4-simplex map of KP type.

Section \ref{Darboux_scheme} deals with the extension to the case of the 4-simplex equation of the ideas presented in \cite{Sokor-Sasha, Sokor-2020} about constructing Yang--Baxter and tetrahedron maps using Darboux transformations. Specifically, we show how refactorisation problems of Darboux matrices can be used to derive 4-simplex maps which can be restricted to parametric 4-simplex maps on invariant leaves. We demonstrate the method by employing a Darboux transformation for the NLS equation and we construct 4-simplex and parametric 4-simplex maps of NLS type.

Finally, in section \ref{conclusions} we summarise the results and present some ideas for future work.

\section{Preliminaries}\label{prelim}
In this section, we explain the relation between the solutions of the local Yang--Baxter equation and the solutions to the functional tetrahedron equation, as well as the relation between the solutions of the local tetrahedron equation and the solutions to the 4-simplex equation.

\subsection{Notation}
Throughout the text:
\begin{itemize}
    \item By $\mathcal{X}$ we denote an arbitrary set, whereas by Latin italic letters (i.e. $x, y, u, v$ etc.) the elements of $\mathcal{X}$, with an exception of the `spectral parameter' which is denoted by the Greek letter $\lambda$. Moreover, by $\End(\mathcal{X})$ we denote the set of maps $\mathcal{X}\rightarrow\mathcal{X}$. Parameters will be denoted by Greek letters and will be complex numbers (i.e. $\alpha, \beta, \gamma, \delta\in\mathbb{C}$).
  
    \item Matrices will be denoted by capital Latin straight letters (i.e. ${\rm A}, {\rm B}, {\rm C}$) etc. Additionally, by ${\rm A}^n_{i_1 i_2\ldots i_m}$ we denote the $n\times n$ extensions of the $m\times m$ ($m<n$) matrix ${\rm A}$, where the elements of ${\rm A}$ are placed at the intersection of the $i_1, i_2\ldots i_m$ rows with the $i_1, i_2\ldots i_m$ columns of matrix ${\rm A}^n_{i_1 i_2\ldots i_m}$. For example, ${\rm A}^3_{12}=\begin{pmatrix} 
 a_{11} &  a_{12} & 0\\ 
a_{21} &  a_{22} & 0\\
0 & 0 & 1
\end{pmatrix}$, ${\rm A}^3_{13}=\begin{pmatrix} 
 a_{11} & 0 &  a_{12} \\
 0 & 1 & 0\\
a_{21} & 0 & a_{22} 
\end{pmatrix}$ and  ${\rm A}^3_{23}=\begin{pmatrix} 
1 & 0 & 0\\
0 & a_{11} &  a_{12} \\ 
0 & a_{21} &  a_{22} \\
\end{pmatrix}$
are a $3\times 3$ extension of matrix ${\rm A}=\begin{pmatrix} 
 a_{11} &  a_{12} \\ 
a_{21} &  a_{22} \\
\end{pmatrix}$.
    
    \item Matrix differential operators are denoted by capital Gothic letters (for instance, $\mathfrak{L}=D_x+\rm{U}$).
\end{itemize}

\subsection{Zamolodchikov's functional tetrahedron VS local Yang--Baxter equation}
A map $T\in\End(\mathcal{X}^3)$, namely
\begin{equation}\label{Tetrahedron_map}
 T:(x,y,z)\mapsto (u(x,y,z),v(x,y,z),w(x,y,z)),
\end{equation}
is called a \textit{tetrahedron map} if it satisfies the \textit{functional tetrahedron} (or Zamolodchikov's tetrahedron) equation
\begin{equation}\label{Tetrahedron-eq}
    T^{123}\circ T^{145} \circ T^{246}\circ T^{356}=T^{356}\circ T^{246}\circ T^{145}\circ T^{123}.
\end{equation}
Functions $T^{ijk}\in\End(\mathcal{X}^6)$, $i,j=1,\ldots 6,~i< j<k$, in \eqref{Tetrahedron-eq} are maps that act as map $T$ on the $ijk$ terms of the Cartesian product $\mathcal{X}^6$ and trivially on the others. For instance,
$$
T^{246}(x,y,z,r,s,t)=(x,u(y,r,t),z,v(y,r,t),s,w(y,r,t)).
$$

Furthermore, if we assign the complex parameters $\alpha$, $\beta$ and $\gamma$ to the variables $x$, $y$ and $z$, respectively, we define a map $T\in\End[(\mathcal{X}\times\mathbb{C})^3]$, namely $T:((x,a),(y,b),(z,c))\mapsto ((u(x,y,z),\alpha),(v(x,y,z),\beta),(w(x,y,z),\gamma))$ which we denote for simplicity as
\begin{equation}\label{Par-Tetrahedron_map}
 T_{\alpha,\beta,\gamma}:(x,y,z)\mapsto (u_{\alpha,\beta,\gamma}(x,y,z),v_{\alpha,\beta,\gamma}(x,y,z),w_{\alpha,\beta,\gamma}(x,y,z)).
\end{equation}
Map \eqref{Par-Tetrahedron_map} is called a \textit{parametric tetrahedron map} if it satisfies the \textit{parametric functional tetrahedron equation}
\begin{equation}\label{Par-Tetrahedron-eq}
    T^{123}_{\alpha,\beta,\gamma}\circ T^{145}_{\alpha,\delta,\epsilon} \circ T^{246}_{\beta,\delta,\zeta}\circ T^{356}_{\gamma,\epsilon,\zeta}=T^{356}_{\gamma,\epsilon,\zeta}\circ T^{246}_{\beta,\delta,\zeta}\circ T^{145}_{\alpha,\delta,\epsilon}\circ T^{123}_{\alpha,\beta,\gamma}.
\end{equation}

Now, let ${\rm L}={\rm L}(x;\kappa)$ be a matrix depending on a variable $x\in\mathcal{X}$ and a parameter $\kappa\in\mathbb{C}$ of the form
\begin{equation}\label{matrix-L}
   {\rm L}(x,\kappa)= \begin{pmatrix} 
a(x,\kappa) & b(x,\kappa)\\ 
c(x,\kappa) & d(x,\kappa)
\end{pmatrix},
\end{equation}
where its entries $a, b, c$ and $d$ are scalar functions of $x$ and $k$. Let ${\rm L}^3_{ij}$, $i,j=1,2, 3$, $i\neq j$, be the $3\times 3$ extensions of matrix \eqref{matrix-L},  defined by
{\small
\begin{equation}\label{Lij-mat}
   {\rm L}^3_{12}(x;\kappa)=\begin{pmatrix} 
 a(x,\kappa) &  b(x,\kappa) & 0\\ 
c(x,\kappa) &  d(x,\kappa) & 0\\
0 & 0 & 1
\end{pmatrix},\quad
 {\rm L}^3_{13}(x;\kappa)= \begin{pmatrix} 
 a(x,\kappa) & 0 & b(x,\kappa)\\ 
0 & 1 & 0\\
c(x,\kappa) & 0 & d(x,\kappa)
\end{pmatrix}, \quad
 {\rm L}^3_{23}(x;\kappa)=\begin{pmatrix} 
   1 & 0 & 0 \\
0 & a(x,\kappa) & b(x,\kappa)\\ 
0 & c(x,\kappa) & d(x,\kappa)
\end{pmatrix},
\end{equation}
}
where ${\rm L}^3_{ij}={\rm L}^3_{ij}(x,\kappa)$, $i,j=1,2,3$.

The following matrix trifactorisation problem
\begin{equation}\label{Lax-Tetra}
    {\rm L}^3_{12}(u;a){\rm L}^3_{13}(v;b){\rm L}^3_{23}(w;c)= {\rm L}^3_{23}(z;c){\rm L}^3_{13}(y;b){\rm L}^3_{12}(x;a),
\end{equation}
where matrices ${\rm L}^3_{ij}$ are defined as in \eqref{Lij-mat}, is the Maillet--Nijhoff equation \cite{Nijhoff} in Korepanov's form, which appears in the literature as the \textit{local Yang--Baxter} equation.

Now, if a map of the form \eqref{Par-Tetrahedron_map} satisfies the local Yang--Baxter equation \eqref{Lax-Tetra}, then this map is a possible tetrahedron map, and equation \eqref{Lax-Tetra} is called its \textit{Lax representation}. In this paper, we consider the case where $a(x,\kappa), b(x,\kappa), c(x,\kappa)$ and $d(x,\kappa)$ in \eqref{Lij-mat} are scalar functions, however Korepanov studied equation \eqref{Lax-Tetra} the case where $a(x,\kappa), b(x,\kappa), c(x,\kappa)$ and $d(x,\kappa)$ in \eqref{Lij-mat} are matrices \cite{Korepanov}. 

Note that matrix refactorisation problems were systematically used for discrete integrable systems already in 90s by Moser and Veselov in \cite{Moser-Veselov}.

\subsection{Set-theoretical 4-simplex equation VS local tetrahedron equation}\label{4-simplexVSLocalYB}
A map $S\in\End(\mathcal{X}^4)$, namely
\begin{equation}\label{4-simplex-map}
 S:(x,y,z,t)\mapsto (u(x,y,z,t),v(x,y,z,t),w(x,y,z,t),r(x,y,z,t)),
\end{equation}
is called a \textit{4-simplex map} if it satisfies the \textit{set-theoretical 4-simplex}  (or Bazhanov--Stroganov) equation
\begin{equation}\label{4-simplex-eq}
    S^{1234}\circ S^{1567} \circ S^{2589}\circ S^{368,10} \circ S^{479,10}=S^{479,10}\circ S^{368,10}\circ S^{2589}\circ S^{1567} \circ S^{1234}.
\end{equation}
Functions $S^{ijkl}\in\End(\mathcal{X}^{10})$, $i,j,k,l=1,\ldots 10,~i< j<k<l$, in \eqref{4-simplex-eq} are maps that act as map $S$ on the $ijkl$ terms of the Cartesian product $\mathcal{X}^{10}$ and trivially on the others. For instance,
\begin{align*}
S^{1567}&(x_1,x_2,\ldots x_{10})=\\
&(u(x_1,x_5,x_6,x_7),x_2,x_3,x_4,v(x_1,x_5,x_6,x_7),w(x_1,x_5,x_6,x_7),r(x_1,x_5,x_6,x_7),x_8,x_9,x_{10}).
\end{align*}

Furthermore, if we assign the complex parameters $\alpha$, $\beta$, $\gamma$ and $\delta$ to the variables $x$, $y$, $z$ and $t$, respectively, we define a map $S\in\End[(\mathcal{X}\times\mathbb{C})^4]$, namely $S:((x,\alpha),(y,\beta),(z,\gamma),(t,\delta))\mapsto ((u(x,y,z,t),\alpha),(v(x,y,z,t),\beta),(w(x,y,z,t),\gamma),(r(x,y,z,t),\delta))$ which we denote for simplicity as
\begin{equation}\label{Par-4-simplex-map}
 S_{\alpha,\beta,\gamma,\delta}:(x,y,z,t)\mapsto (u_{\alpha,\beta,\gamma,\delta}(x,y,z,t),v_{\alpha,\beta,\gamma,\delta}(x,y,z,t),w_{\alpha,\beta,\gamma,\delta}(x,y,z,t),r_{\alpha,\beta,\gamma,\delta}(x,y,z,t)).
\end{equation}
Map \eqref{Par-Tetrahedron_map} is called a \textit{parametric 4-simplex map} if it satisfies the \textit{parametric 4-simplex equation}
\begin{equation}\label{Par-4-simplex-eq}
     S^{1234}_{\alpha,\beta,\gamma,\delta}\circ S^{1567}_{\alpha,\epsilon,\zeta,\theta} \circ S^{2589}_{\beta,\epsilon,\kappa,\mu}\circ S^{368,10}_{\gamma,\zeta, \kappa, \nu} \circ S^{479,10}_{\delta,\theta,\mu,\nu}=S^{479,10}_{\delta,\theta,\mu,\nu}\circ S^{368,10}_{\gamma,\zeta, \kappa, \nu}\circ S^{2589}_{\beta,\epsilon,\kappa,\mu}\circ S^{1567}_{\alpha,\epsilon,\zeta,\theta} \circ S^{1234}_{\alpha,\beta,\gamma,\delta}.
\end{equation}

Now, let ${\rm L}={\rm L}(x;\kappa)$ be a $3\times 3$ square matrix depending on a variable $x\in\mathcal{X}$ and a parameter $\kappa\in\mathbb{C}$ of the form
\begin{equation}\label{matrix-L-simplex}
   {\rm L}(x;\kappa)= \begin{pmatrix} 
a(x,\kappa) & b(x,\kappa) & c(x,\kappa)\\ 
d(x,\kappa) & e(x,\kappa) & f(x,\kappa)\\
k(x,\kappa) & l(x,\kappa) & m(x,\kappa)
\end{pmatrix},
\end{equation}
where its entries $a, b, c, d, e, f, k, l$ and $m$ are scalar functions of $x$ and $k$. Let ${\rm L}^6_{ijk}(x;\kappa)$, $i,j,k=1,\ldots 6$, $i< j<k$, be the $6\times 6$ extensions of matrix \eqref{matrix-L-simplex},  defined by
{\small
\begin{subequations}\label{6x6-extensions}
\begin{align}
    & {\rm L}^6_{123}(x;\kappa)= \begin{pmatrix} 
a(x,\kappa) & b(x,\kappa) & c(x,\kappa) & 0 & 0 & 0\\ 
d(x,\kappa) & e(x,\kappa) & f(x,\kappa) & 0 & 0 & 0\\
k(x,\kappa) & l(x,\kappa) & m(x,\kappa) & 0 & 0 & 0\\
0 & 0 & 0 & 1 & 0 & 0 \\
0 & 0 & 0 & 0 & 1 & 0 \\
0 & 0 & 0 & 0 & 0 & 1 \\
\end{pmatrix}, \quad
{\rm L}^6_{145}(x;\kappa)= \begin{pmatrix} 
a(x,\kappa) & 0 & 0 & b(x,\kappa) & c(x,\kappa) & 0\\ 
0 & 1 & 0 & 0 & 0 & 0 \\
0 & 0 & 1 & 0 & 0 & 0 \\
d(x,\kappa) & 0 & 0 & e(x,\kappa) & f(x,\kappa) & 0\\
k(x,\kappa) & 0 & 0 & l(x,\kappa) & m(x,\kappa) & 0\\
0 & 0 & 0 & 0 & 0 & 1 \\
\end{pmatrix},\\
& {\rm L}^6_{246}(x;\kappa)= \begin{pmatrix} 
1 & 0 & 0 & 0 & 0 & 0 \\
0 & a(x,\kappa) & 0 & b(x,\kappa) & 0 & c(x,\kappa)\\ 
0 & 0 & 1 & 0 & 0 & 0 \\
0 & d(x,\kappa) & 0 & e(x,\kappa) & 0 & f(x,\kappa)\\
0 & 0 & 0 & 0 & 1 & 0 \\
0 & k(x,\kappa) & 0 & l(x,\kappa) & 0 & m(x,\kappa)
\end{pmatrix},\quad 
 {\rm L}^6_{356}(x;\kappa)= \begin{pmatrix} 
 1 & 0 & 0 & 0 & 0 & 0 \\
 0 & 1 & 0 & 0 & 0 & 0 \\
0 & 0 & a(x,\kappa) & 0 & b(x,\kappa) & c(x,\kappa)\\ 
0 & 0 & 0 & 1 & 0 & 0 \\
0 & 0 & d(x,\kappa) & 0 & e(x,\kappa) & f(x,\kappa)\\
0 & 0 & k(x,\kappa) & 0 & l(x,\kappa) & m(x,\kappa)
\end{pmatrix}.
\end{align}
\end{subequations}
}

We call the following matrix four-factorisation problem
\begin{equation}\label{local-tetrahedron}
    {\rm L}^6_{123}(u;\alpha){\rm L}^6_{145}(v;\beta){\rm L}^6_{246}(w;\gamma){\rm L}^6_{356}(r;\delta)={\rm L}^6_{356}(t;\delta){\rm L}^6_{246}(z;\gamma){\rm L}^6_{145}(y;\beta){\rm L}^6_{123}(x;\alpha)
\end{equation}
\textit{local tetrahedron equation}. The local tetrahedron equation is a generator of potential solutions to the 4-simplex equation. If map \eqref{Par-4-simplex-map} satisfies equation \eqref{local-tetrahedron}, then the matrix refactorisation problem \eqref{local-tetrahedron} is called a \textit{Lax representation} for map \eqref{Par-4-simplex-map}.

\section{A method for constructing 4-simplex maps}\label{Extension method}
In this section, we present a method for constructing 4-simplex maps by generalising the generators of tetrahedron maps. Then, we apply this method to the list of generators of the tetrahedron maps of ferro-electric type classified by Sergeev \cite{Sergeev} and studied by Kashaev, Korepanov and Sergeev in \cite{Kashaev-Sergeev} as well as to the generator of a KP-type tetrahedron map appeared in \cite{Dimakis}. As a result, we construct solutions to the 4-simplex equations as nontrivial extensions of tetrahedron maps. 

\subsection{Formulation of the method}\label{method}
Any tetrahedron map can be extended to a 4-simplex map in the following way.

\begin{proposition}\label{trivial_extension}
Let $T\in\End(\mathcal{X}^3)$, $T:(x,y,z)\rightarrow (u(x,y,z),v(x,y,z),w(x,y,z))$, be a tetrahedron map. Then, the map $S\in\mathcal(X)^4$, $S:(x,y,z,t)\rightarrow (u(x,y,z),v(x,y,z),w(x,y,z),t)$, is a 4-simplex map.
\end{proposition}
\begin{proof}
It can be readily verified by substitution of $S:(x,y,z,t)\rightarrow (u(x,y,z),v(x,y,z),w(x,y,z),t)$ into the 4-simplex equation.
\end{proof}

We call all 4-simplex maps which are obtained from tetrahedron maps as in Proposition \ref{trivial_extension} \textit{trivial solutions} to the 4-simplex equation. In order to construct nontrivial solutions to the 4-simplex equation, one may extend the generator of tetrahedron maps. 

In particular, let us consider a tetrahedron map \eqref{Tetrahedron_map} with  Lax representation \eqref{Lax-Tetra}. If we consider the $3\times 3$ extension of matrix \eqref{matrix-L}, namely matrix
\begin{equation}\label{M-matrix}
{\rm M}(x;\kappa)\equiv{\rm L}^3_{12}(x;\kappa)= \begin{pmatrix} 
a(x,\kappa) & b(x,\kappa) & 0\\ 
c(x,\kappa) & d(x,\kappa) & 0 \\
0 & 0 & 1
\end{pmatrix}
\end{equation}
and substitute it to the local tetrahedron equation \eqref{local-tetrahedron}, we will obtain the 4-simplex map
$$S:(x,y,z,t)\rightarrow (u(x,y,z),v(x,y,z),w(x,y,z),t),$$
namely a trivial solution to the 4-simplex equation with Lax representation 
$$
  {\rm M}^6_{123}(u;\alpha){\rm M}^6_{145}(v;\beta){\rm M}^6_{246}(w;\gamma){\rm M}^6_{356}(r;\delta)={\rm M}^6_{356}(t;\delta){\rm M}^6_{246}(z;\gamma){\rm M}^6_{145}(y;\beta){\rm M}^6_{123}(x;\alpha),
$$
where ${\rm M}^6_{ijk}$, $i,j,k=1,\ldots 6$, $i<j<k$, are the $6\times 6$ generalisations of matrix \eqref{M-matrix}. In order to derive a nontrivial 4-simplex map, one must introduce a new variable. That is, instead of considering matrix \eqref{M-matrix} as a $3\times 3$ extension of matrix \eqref{matrix-L}, we introduce the following matrix
\begin{equation}\label{K-matrix}
{\rm K}(x_1,x_2;\kappa)= \begin{pmatrix} 
a(x_1,\kappa) & b(x_1,\kappa) & 0\\ 
c(x_1,\kappa) & d(x_1,\kappa) & 0 \\
0 & 0 & x_2
\end{pmatrix}
\end{equation}
such that for $x_2\rightarrow 1$, we have ${\rm K}(x_1,x_2;\kappa)\rightarrow {\rm M}(x_1;\kappa)$.

Now, substituting ${\rm K}(x_1,x_2;\kappa)$ to the local tetrahedron equation
\begin{align*}
  &{\rm K}^6_{123}(u_1,u_2;\alpha){\rm K}^6_{145}(v_1,v_2;\beta){\rm K}^6_{246}(w_1,w_2;\gamma){\rm K}^6_{356}(r_1,r_2;\delta)=\\
  &{\rm K}^6_{356}(t_1,t_2;\delta){\rm K}^6_{246}(z_1,z_2;\gamma){\rm K}^6_{145}(y_1,y_2;\beta){\rm K}^6_{123}(x_1,x_2;\alpha),
\end{align*}
we aim to obtain a correspondence, which for particular values of the free variables will define 4-simplex maps. As we will see in certain applications, this correspondence does not define 4-simplex maps for arbitrary choice of the free variables. However, for certain choice of the free variables these correspondences define 4-simplex maps.

\begin{remark}\normalfont
    As explained in Section \ref{4-simplexVSLocalYB}, $4$-simplex maps can be generated  by $3\times 3$ matrices \eqref{matrix-L-simplex} by substitution of the latter to the local tetrahedron equation \eqref{local-tetrahedron}. Matrix \eqref{M-matrix} is a  trivial $3\times 3$ extension of matrix \eqref{matrix-L}, and it generates trivial extensions of tetrahedron maps, whereas the simplest nontrivial $3\times 3$ extension of matrix \eqref{matrix-L} can be obtained by  replacing $1$ in \eqref{M-matrix} by a variable $x_2$, namely consider matrix \eqref{K-matrix}. Our motivation for this choice is to demonstrate that even the simplest  $3\times 3$ nontrivial extension  of matrix \eqref{M-matrix} leads to new interesting $4$-simplex maps as nontrivial extensions of  tetrahedron maps. However, this is not  the only $3\times 3$ extension one can consider.
\end{remark}

\subsection{Kashaev--Korepanov--Sergeev type 4-simplex maps}\label{KKS maps}
We apply the method presented in the previous section to the generators of the maps in Sergeev's classification \cite{Sergeev, Kashaev-Sergeev} and we construct new 4-simplex maps. These maps can be considered as extensions of the Kashaev--Korepanov--Sergeev tetrahedron maps.

\subsubsection{Case a.}
Consider the $k$-parametric family of tetrahedron maps
\begin{equation}\label{Sergeev-a}
    T:(x,y,z)\rightarrow \left(x,\frac{y-xz}{\kappa},z\right),
\end{equation}
with Lax representation
$$
 {\rm L}^3_{12}(u;k){\rm L}^3_{13}(v;k){\rm L}^3_{23}(w;k)= {\rm L}^3_{23}(z;k){\rm L}^3_{13}(y;k){\rm L}^3_{12}(x;k),
$$
where ${\rm L}^3_{ij}$, $i,j=1,2,3$, $i<j$, are the $3\times 3$ generalisations of matrix
\begin{equation}\label{Lax-Sergeev-a}
    {\rm L}(x;\kappa)=\begin{pmatrix} 
1 & x\\ 
0 & \kappa
\end{pmatrix},
\end{equation}
and let ${\rm M}(x;\kappa)\equiv {\rm L}^3_{12}(x;\kappa)=\begin{pmatrix} 
1 & x & 0\\ 
0 & \kappa & 0\\
0 & 0 & 1
\end{pmatrix}$
be its $3\times 3$ extension.

Now consider the following generalisation of matrix ${\rm M}(x;\kappa)$:
\begin{equation}\label{Gen-Lax-Sergeev-a}
    {\rm K}(x_1,x_2;\kappa)=\begin{pmatrix} 
1 & x_1 & 0\\ 
0 & \kappa & 0 \\
0 & 0 & x_2
\end{pmatrix},
\end{equation}
such that for $x_2\rightarrow 1$, ${\rm K}(x_1,x_2;\kappa)\rightarrow {\rm M}(x_1;\kappa)$. We consider the $6\times 6$ extensions of matrix ${\rm K}(x_1,x_2;\kappa)$, namely the following
\begin{align*}
    & {\rm K}^6_{123}(x_1,x_2;\kappa)= \begin{pmatrix} 
1 & x_1 & 0 & 0 & 0 & 0\\ 
0 & \kappa & 0 & 0 & 0 & 0\\
0 & 0 & x_2 & 0 & 0 & 0\\
0 & 0 & 0 & 1 & 0 & 0 \\
0 & 0 & 0 & 0 & 1 & 0 \\
0 & 0 & 0 & 0 & 0 & 1 \\
\end{pmatrix}, \quad
{\rm K}^6_{145}(x_1,x_2;\kappa)= \begin{pmatrix} 
1 & 0 & 0 & x_1 & 0 & 0\\ 
0 & 1 & 0 & 0 & 0 & 0 \\
0 & 0 & 1 & 0 & 0 & 0 \\
0 & 0 & 0 & \kappa & 0 & 0\\
0 & 0 & 0 & 0 & x_2 & 0\\
0 & 0 & 0 & 0 & 0 & 1 \\
\end{pmatrix},\\
& {\rm K}^6_{246}(x_1,x_2;\kappa)= \begin{pmatrix} 
1 & 0 & 0 & 0 & 0 & 0 \\
0 & 1 & 0 & x_1 & 0 & 0\\ 
0 & 0 & 1 & 0 & 0 & 0 \\
0 & 0 & 0 & \kappa & 0 & 0\\
0 & 0 & 0 & 0 & 1 & 0 \\
0 & 0 & 0 & 0 & 0 & x_2
\end{pmatrix},\quad 
 {\rm K}^6_{356}(x_1,x_2;\kappa)= \begin{pmatrix} 
 1 & 0 & 0 & 0 & 0 & 0 \\
 0 & 1 & 0 & 0 & 0 & 0 \\
0 & 0 & 1 & 0 & x_1 & 0\\ 
0 & 0 & 0 & 1 & 0 & 0 \\
0 & 0 & 0 & 0 & \kappa & 0\\
0 & 0 & 0 & 0 & 0 & x_2
\end{pmatrix},
\end{align*}
and substitute them to the local tetrahedron equation
\begin{align*}
  &{\rm K}^6_{123}(u_1,u_2,\kappa){\rm K}^6_{145}(v_1,v_2;\kappa){\rm K}^6_{246}(w_1,w_2;\kappa){\rm K}^6_{356}(r_1,r_2;\kappa)=\\
  &{\rm K}^6_{356}(t_1,t_2;\kappa){\rm K}^6_{246}(z_1,z_2;\kappa){\rm K}^6_{145}(y_1,y_2;\kappa){\rm K}^6_{123}(x_1,x_2;\kappa).
\end{align*}
The above implies the following correspondence:
\begin{equation}\label{corr-serg-a}
  u_1=x_1,\quad u_2=x_2,\quad v_1=\frac{y_1-x_1z_1}{k},\quad v_2=y_2,\quad w_1=z_1,\quad r_1=\frac{t_1y_2}{x_2},\quad r_2=\frac{t_2z_2}{w_2}.  
\end{equation}

For the choices of the free variable $w_2=t_2$ and $w_2=y_2$ the above correspondence defines 4-simplex maps. In particular, we have the following.

\begin{proposition}\label{case-a}
The following maps
\begin{equation}\label{4-simplex-a}
   S_1: (x_1,x_2,y_1,y_2,z_1,z_2,t_1,t_2)\rightarrow \left(x_1,x_2,\frac{y_1-x_1z_1}{k},y_2,z_1,t_2,\frac{t_1y_2}{x_2},z_2\right),
\end{equation}
and 
\begin{equation}\label{4-simplex-a-2}
   S_2: (x_1,x_2,y_1,y_2,z_1,z_2,t_1,t_2)\rightarrow \left(x_1,x_2,\frac{y_1-x_1z_1}{k},y_2,z_1,y_2,\frac{t_1y_2}{x_2},\frac{t_2z_2}{y_2}\right),
\end{equation}
are eight-dimensional, noninvolutive 4-simplex maps which share the same invariants $I_1=x_1$, $I_2=x_2$, $I_3=y_2$ and $I_4=z_1$. Moreover, map $S_1$ is birational.
\end{proposition}
\begin{proof}
Maps \eqref{4-simplex-a} and \eqref{4-simplex-a-2} follow after substitution of $w_2=t_2$ and $w_2=y_2$, respectively, to \eqref{corr-serg-a}. The 4-simplex property can be verified with straightforward substitution to the 4-simplex equation \eqref{4-simplex-eq}. Now, since $r_1\circ S_1 = r_1\circ S_2 = \frac{t_1y_2^2}{x_2^2}\neq t_1$, it follows that $S_i\circ S_i\neq\id$ $i=1,2$, thus maps \eqref{4-simplex-a} and \eqref{4-simplex-a-2} are noninvolutive. The invariants are obvious.

Finally, the inverse of map \eqref{4-simplex-a} is given by:
$$
S_1^{-1}(u_1,u_2,v_1,v_2,w_1,w_2,r_1,r_2)\rightarrow \left(u_1,u_2,kv_1+u_1w_1,v_2,w_1,r_2,\frac{r_1u_2}{v_2},w_2 \right),
$$
therefore map \eqref{4-simplex-a} is birational.
\end{proof}

\begin{remark}\normalfont In proposition \ref{case-a} we derived 4-simplex maps by fixing the free variable $w_2$ in correspondence \eqref{corr-serg-a}. However, correspondence \eqref{corr-serg-a} does not define 4-simplex maps for arbitrary value of the variable $w_2$. For instance, for the choice $w_2=x_2$, correspondence \eqref{corr-serg-a} defines the following map
$$
(x_1,x_2,y_1,y_2,z_1,z_2,t_1,t_2)\rightarrow \left(x_1,x_2,\frac{y_1-x_1z_1}{k},y_2,z_1,x_2,\frac{t_1y_2}{x_2},\frac{t_2z_2}{x_2}\right),
$$
which is not a 4-simplex map.
\end{remark}

\subsubsection{Case b.}
Consider the $k$-parametric family of tetrahedron maps
\begin{equation}\label{Sergeev-b}
    T:(x,y,z)\rightarrow \left(\frac{xy}{y+xz},\frac{xz}{k},\frac{y+xz}{x}\right),
\end{equation}
with Lax representation
$$
 {\rm L}^3_{12}(u;k){\rm L}^3_{13}(v;k){\rm L}^3_{23}(w;k)= {\rm L}^3_{23}(z;k){\rm L}^3_{13}(y;k){\rm L}^3_{12}(x;k),
$$
where ${\rm L}^3_{ij}$, $i,j=1,2,3$, $i<j$, are the $3\times 3$ generalisations of matrix
\begin{equation}\label{Lax-Sergeev-b}
    {\rm L}(x;\kappa)=\begin{pmatrix} 
1 & x\\ 
\frac{\kappa}{x} & 0
\end{pmatrix},
\end{equation}
and let ${\rm M}(x;\kappa)\equiv {\rm L}^3_{12}(x;\kappa)=\begin{pmatrix} 
1 & x & 0\\ 
\frac{\kappa}{x} & 0 & 0\\
0 & 0 & 1
\end{pmatrix}$
be its $3\times 3$ extension.

Now consider the following generalisation of matrix ${\rm M}(x;\kappa)$:
\begin{equation}\label{Gen-Lax-Sergeev-b}
    {\rm K}(x_1,x_2;\kappa)=\begin{pmatrix} 
1 & x_1 & 0\\ 
\frac{\kappa}{x_1} & 0 & 0 \\
0 & 0 & x_2
\end{pmatrix},
\end{equation}
such that for $x_2\rightarrow 1$, ${\rm K}(x_1,x_2;\kappa)\rightarrow {\rm M}(x_1;\kappa)$, and consider the $6\times 6$ extensions of ${\rm K}(x_1,x_2;\kappa)$, ${\rm K}^6_{ijk}(x_1,x_2;\kappa)$, as in \eqref{6x6-extensions}. Then, we substitute ${\rm K}^6_{ijk},$ $i=1,\ldots, 6$, $i<j<k$ to the local tetrahedron equation
\begin{align*}
  &{\rm K}^6_{123}(u_1,u_2,\kappa){\rm K}^6_{145}(v_1,v_2;\kappa){\rm K}^6_{246}(w_1,w_2;\kappa){\rm K}^6_{356}(r_1,r_2;\kappa)=\\
  &{\rm K}^6_{356}(t_1,t_2;\kappa){\rm K}^6_{246}(z_1,z_2;\kappa){\rm K}^6_{145}(y_1,y_2;\kappa){\rm K}^6_{123}(x_1,x_2;\kappa).
\end{align*}
The above implies the following correspondence:
\begin{equation}\label{corr-serg-b}
  u_1=\frac{x_1y_1}{y_1+x_1z_1},\quad u_2=x_2,\quad v_1=\frac{x_1z_1}{k},\quad v_2=y_2,\quad w_1=\frac{y_1+x_1z_1}{x_1},\quad r_1=\frac{t_1y_2}{x_2},\quad r_2=\frac{t_2z_2}{w_2}.  
\end{equation}

The above correspondence defines 4-simplex maps for the choices of the free variable $w_2=t_2$ and $w_2=y_2$. Specifically, we have the following.

\begin{proposition}\label{case-b}
The following maps
\begin{equation}\label{4-simplex-b}
   S_1: (x_1,x_2,y_1,y_2,z_1,z_2,t_1,t_2)\rightarrow \left(\frac{x_1y_1}{y_1+x_1z_1}, x_2, \frac{x_1z_1}{k}, y_2, \frac{y_1+x_1z_1}{x_1}, t_2, \frac{t_1y_2}{x_2}, z_2\right),
\end{equation}
with invariants $I_1=x_2$, $I_2=y_2$, $I_3=t_2z_2$ and $I_4=t_2+z_2$, and
\begin{equation}\label{4-simplex-b-2}
   S_2: (x_1,x_2,y_1,y_2,z_1,z_2,t_1,t_2)\rightarrow \left(\frac{x_1y_1}{y_1+x_1z_1}, x_2, \frac{x_1z_1}{k}, y_2, \frac{y_1+x_1z_1}{x_1}, y_2, \frac{t_1y_2}{x_2}, \frac{t_2z_2}{y_2}\right),
\end{equation}
with invariants $I_1=x_2$, $I_2=y_2$, $I_3=t_2z_2$ and $I_4=x_2y_2z_2t_2$, are eight-dimensional, noninvolutive 4-simplex maps. Moreover, map $S_1$ is birational.
\end{proposition}
\begin{proof}
Maps \eqref{4-simplex-b} and \eqref{4-simplex-b-2} follow after substitution of $w_2=t_2$ and $w_2=y_2$, respectively, to \eqref{corr-serg-b}. The 4-simplex property can be verified with straightforward substitution to the 4-simplex equation \eqref{4-simplex-eq}. Now, since $u_1\circ S_1 = u_1\circ S_2 = \frac{x_1^2y_1z_1}{(y_1+x_1z_1)(ky_1+x_1z_1)}\neq x_1$, it follows that $S_i\circ S_i\neq\id$ $i=1,2$, thus maps \eqref{4-simplex-b} and \eqref{4-simplex-b-2} are noninvolutive. For the invariants we have that, in view of \eqref{4-simplex-b} and \eqref{4-simplex-b-2}, $u_2=x_2$, $v_2=y_2$, $r_2w_2=z_2t_2$ and $u_2v_2w_2r_2=x_2y_2z_2t_2$.

Finally, the inverse of map \eqref{4-simplex-b} is given by:
$$
S_1^{-1}(u_1,u_2,v_1,v_2,w_1,w_2,r_1,r_2)\rightarrow \left(\frac{kv_1+u_1w_1}{w_1},u_2,u_1w_1,v_2,\frac{kv_1w_1}{kv_1+u_1w_1},r_2,\frac{r_1u_2}{v_2},w_2\right),
$$
therefore map \eqref{4-simplex-b} is birational.
\end{proof}

\subsubsection{Case c.}
Consider the tetrahedron map
\begin{equation}\label{Sergeev-c}
    T:(x,y,z)\rightarrow \left(\frac{xy}{xy-xz+z},xy+z-xz,z\right),
\end{equation}
with Lax representation
$$
 {\rm L}^3_{12}(u){\rm L}^3_{13}(v){\rm L}^3_{23}(w)= {\rm L}^3_{23}(z){\rm L}^3_{13}(y){\rm L}^3_{12}(x),
$$
where ${\rm L}^3_{ij}(x)$, $i,j=1,2,3$, $i<j$, are the $3\times 3$ generalisations of matrix
\begin{equation}\label{Lax-Sergeev-c}
    {\rm L}(x)=\begin{pmatrix} 
x & 0\\ 
1-x & 1
\end{pmatrix},
\end{equation}
and let ${\rm M}(x)\equiv {\rm L}^3_{12}(x)=\begin{pmatrix} 
x & 0 & 0\\ 
1-x & 1 & 0\\
0 & 0 & 1
\end{pmatrix}$
be its $3\times 3$ extension.

Now, consider the following generalisation of matrix ${\rm M}(x)$:
\begin{equation}\label{Gen-Lax-Sergeev-c}
    {\rm K}(x_1,x_2)=\begin{pmatrix} 
x_1 & 0 & 0\\ 
1-x_1 & 1 & 0\\
0 & 0 & x_2
\end{pmatrix},
\end{equation}
such that for $x_2\rightarrow 1$, ${\rm K}(x_1,x_2)\rightarrow {\rm M}(x_1)$, and substitute ${\rm K}(x_1,x_2)$ to the local tetrahedron equation
\begin{align*}
  &{\rm K}^6_{123}(u_1,u_2){\rm K}^6_{145}(v_1,v_2){\rm K}^6_{246}(w_1,w_2){\rm K}^6_{356}(r_1,r_2)=\\
  &{\rm K}^6_{356}(t_1,t_2){\rm K}^6_{246}(z_1,z_2){\rm K}^6_{145}(y_1,y_2){\rm K}^6_{123}(x_1,x_2),
\end{align*}
where ${\rm K}^6_{ijk}$, $i=1,\ldots, 6$, $i<j<k$, are the $6\times 6$ extensions of matrix ${\rm K}(x_1,x_2)$ as in \eqref{6x6-extensions}. The above implies the following correspondence:
\begin{subequations}
\begin{align}\label{corr-serg-c}
  &u_1=\frac{x_1y_1}{x_1y_1+z_1-x_1z_1},\quad u_2=\frac{t_1x_2y_2}{y_2+t_1x_2-x_2},\quad v_1=x_1y_1+z_1-x_1z_1,\\ 
  & v_2=y_2,\quad w_1=z_1,\quad r_1=\frac{y_2+t_1x_2-x_2}{y_2},\quad r_2=\frac{t_2z_2}{w_2}.  
\end{align}
\end{subequations}

The above correspondence defines a 4-simplex map for the choice of the free variable $w_2=t_2$. In particular, we have the following.

\begin{proposition}\label{case-c}
The following map\small
\begin{equation}\label{4-simplex-c}
(x_1,x_2,y_1,y_2,z_1,z_2,t_1,t_2)\stackrel{S}{\rightarrow} \left(\frac{x_1y_1}{x_1y_1+z_1-x_1z_1}, \frac{t_1x_2y_2}{y_2+t_1x_2-x_2}, x_1y_1+z_1-x_1z_1, y_2, z_1, t_2, \frac{y_2+t_1x_2-x_2}{y_2}, z_2\right),
\end{equation}
\normalfont
with invariants $I_1=z_1$, $I_2=x_1y_1$, $I_3=t_2z_2$ and $I_4=t_2+z_2$ is an eight-dimensional, noninvolutive 4-simplex map. 
\end{proposition}
\begin{proof}
Maps \eqref{4-simplex-c} is obtained after substitution of $w_2=t_2$ to \eqref{corr-serg-c}. The 4-simplex property can be verified with straightforward substitution to the 4-simplex equation \eqref{4-simplex-eq}. Now, since $u_2\circ S = \frac{t_1^2x_2y_2}{(t_1^2-1)x_2+y_2}\neq x_2$, it follows that $S\circ S\neq\id$, thus map \eqref{4-simplex-c} is noninvolutive. Regarding the invariants, we have that $w_1\stackrel{\eqref{4-simplex-c}}{=}z_1$, $u_1v_1\stackrel{\eqref{4-simplex-c}}{=}x_1y_1$, $w_2r_2\stackrel{\eqref{4-simplex-c}}{=}z_2t_2$ and $w_2+t_2\stackrel{\eqref{4-simplex-c}}{=}r_2+t_2$.
\end{proof}

\subsubsection{Case d.}
Consider the tetrahedron map
\begin{equation}\label{Sergeev-d}
    T:(x,y,z)\rightarrow \left(\frac{xy}{x+z-xz},x+z-xz,\frac{(1-x)yz}{x+z-xy-xz}\right),
\end{equation}
with Lax representation
$$
 {\rm L}^3_{12}(u){\rm L}^3_{13}(v){\rm L}^3_{23}(w)= {\rm L}^3_{23}(z){\rm L}^3_{13}(y){\rm L}^3_{12}(x),
$$
where ${\rm L}^3_{ij}(x)$, $i,j=1,2,3$, $i<j$, are the $3\times 3$ generalisations of matrix
\begin{equation}\label{Lax-Sergeev-d}
    {\rm L}(x)=\begin{pmatrix} 
x & 1\\ 
1-x & 0
\end{pmatrix},
\end{equation}
and let ${\rm M}(x)\equiv {\rm L}^3_{12}(x)=\begin{pmatrix} 
x & 1 & 0\\ 
1-x & 0 & 0\\
0 & 0 & 1
\end{pmatrix}$
be its $3\times 3$ extension.

Now, consider the following generalisation of matrix ${\rm M}(x;\kappa)$:
\begin{equation}\label{Gen-Lax-Sergeev-d}
    {\rm K}(x_1,x_2)=\begin{pmatrix} 
x_1 & 1 & 0\\ 
1-x_1 & 0 & 0\\
0 & 0 & x_2
\end{pmatrix},
\end{equation}
such that for $x_2\rightarrow 1$, ${\rm K}(x_1,x_2)\rightarrow {\rm M}(x_1)$, and substitute ${\rm K}(x_1,x_2)$ to the local tetrahedron equation
\begin{align*}
  &{\rm K}^6_{123}(u_1,u_2){\rm K}^6_{145}(v_1,v_2){\rm K}^6_{246}(w_1,w_2){\rm K}^6_{356}(r_1,r_2)=\\
  &{\rm K}^6_{356}(t_1,t_2){\rm K}^6_{246}(z_1,z_2){\rm K}^6_{145}(y_1,y_2){\rm K}^6_{123}(x_1,x_2),
\end{align*}
where ${\rm K}^6_{ijk}$, $i=1,\ldots, 6$, $i<j<k$, are the $6\times 6$ extensions of matrix ${\rm K}(x_1,x_2)$ as in \eqref{6x6-extensions}. The above implies the following correspondence:
\begin{subequations}
\begin{align}\label{corr-serg-d}
  &u_1=\frac{x_1y_1}{x_1+z_1-x_1z_1},\quad u_2=y_2,\quad v_1=x_1+z_1-x_1z_1,\\ 
  & v_2=\frac{(1-t_1)x_2y_2}{y_2-t_1x_2},\quad w_1=\frac{(1-x_1)y_1z_1}{x_1+z_1-x_1y_1-x_1z_1},\quad r_1=\frac{t_1x_2}{y_2},\quad r_2=\frac{t_2z_2}{w_2}.  
\end{align}
\end{subequations}

The above correspondence defines a 4-simplex map for the choice of the free variable $w_2=t_2$. In particular, we have the following.

\begin{proposition}\label{case-d}
The following map\small
\begin{equation}\label{4-simplex-d}
(x_1,x_2,y_1,y_2,z_1,z_2,t_1,t_2)\stackrel{S}{\rightarrow} \left(\frac{x_1y_1}{x_1+z_1-x_1z_1}, y_2, x_1+z_1-x_1z_1, \frac{(1-t_1)x_2y_2}{y_2-t_1x_2},\frac{(1-x_1)y_1z_1}{x_1(1-y_1)+z_1(1-x_1)}, t_2, \frac{t_1x_2}{y_2}, z_2\right),
\end{equation}\normalfont
with invariants $I_1=y_2$, $I_2=x_1y_1$, $I_3=t_2z_2$ and $I_4=t_2+z_2$ is an eight-dimensional, noninvolutive 4-simplex map. 
\end{proposition}
\begin{proof}
Maps \eqref{4-simplex-d} is obtained after substitution of $w_2=t_2$ to \eqref{corr-serg-d}. The 4-simplex property can be verified with straightforward substitution to the 4-simplex equation \eqref{4-simplex-eq}. Now, since $u_2\circ S = \frac{t_1^2x_2y_2}{(t_1^2-1)x_2+y_2}\neq x_2$, it follows that $S\circ S\neq\id$, thus map \eqref{4-simplex-d} is noninvolutive. For the invariants, we have that $v_2\stackrel{\eqref{4-simplex-d}}{=}y_2$, $u_1v_1\stackrel{\eqref{4-simplex-d}}{=}x_1y_1$, $w_2r_2\stackrel{\eqref{4-simplex-d}}{=}z_2t_2$ and $w_2+t_2\stackrel{\eqref{4-simplex-d}}{=}r_2+t_2$.
\end{proof}

\subsection{Kadomtsev--Petviashvili type of 4-simplex map}
Let $\mathcal{X}=A$ be an algebraic variety. Consider the KP-type tetrahedron map $T\in\End{A^6}$ given by:
\begin{subequations}\label{KP}
\begin{align}
    x_1 &\mapsto u_1=\frac{x_1y_1}{z_1+x_1(y_1-z_1+z_2-y_1z_2)},\\
    x_2 &\mapsto u_2=\frac{y_2(1-z_1)+x_2y_1(z_2-1)+x_2y_2(z_1-z_2)}{(1-x_2)(z_2-z_1)+(1-z_2)(y_2-x_2y_1)},\\
    y_1 &\mapsto v_1= z_1+x_1(y_1-z_1+z_2-y_1z_2),\\
    y_2 &\mapsto v_2=\frac{\left[z_1+x_1(y_1-z_1+z_2-y_1z_2)\right] \left[(1-x_2)y_2z_1-(y_1-y_2)x_2z_2\right]}{x_1 (y_1-y_2) (z_1-z_2)-z_1(x_2y_1-y_2)},\\
    z_1 &\mapsto w_1=\frac{(x_1-x_2)y_1z_1\left[y_1x_2(z_2-1)+(z_1-z_2)(x_2-1)-y_2(z_2-1)\right]}{\Lambda},\\
    z_2 &\mapsto w_2=\frac{\left[x_1z_2(y_1-y_2)+y_2z_1(x_1-1)\right]\left[y_1x_2(z_2-1)+(z_1-z_2)(x_2-1)-y_2(z_2-1)\right]}{\Lambda},
\end{align}
\end{subequations}
where
\begin{align}
   \Lambda=&y_2z_1(z_1-1)-x_2z_1\left[y_1(z_2-1)+y_2(z_1-z_2)\right]+x_1y_2(z_1-z_2)\left[1-z_1+x_2(z_1-z_2)\right]+\nonumber\\
   &x_1y_1\left[x_2z_2(y_1-y_2)(z_2-1)+(z_2-z_1)(1-x_2z_2)+z_1y_2(x_2-1)(z_2-1)\right],\label{lambda-expr} 
\end{align}
which was constructed by Dimakis and M\"uller-Hoissen \cite{Dimakis} in a study of soliton solutions of vector KP equations. Map \eqref{KP} admits the following Lax representation:
$$
 {\rm L}^3_{12}(u_1,u_2){\rm L}^3_{13}(v_1,v_2){\rm L}^3_{23}(w_1,w_2)= {\rm L}^3_{23}(z_1,z_2){\rm L}^3_{13}(y_1,y_2){\rm L}^3_{12}(x_1,x_2),
$$
where ${\rm L}^3_{ij}$, $i,j=1,2,3$, $i<j$, are the $3\times 3$ generalisations of matrix
\begin{equation}\label{Lax-KP}
    {\rm L}(x_1,x_2)=\begin{pmatrix} 
x_1 & x_2\\ 
1-x_1 & 1-x_2
\end{pmatrix},
\end{equation}
and let ${\rm M}(x_1,x_2)\equiv {\rm L}^3_{12}(x_1,x_2)=\begin{pmatrix} 
x_1 & x_2 & 0\\ 
1-x_1 & 1-x_2 & 0\\
0 & 0 & 1
\end{pmatrix}$
be its $3\times 3$ extension.

Now consider the following generalisation of matrix ${\rm M}(x;\kappa)$:
\begin{equation}\label{Gen-Lax-KP}
    {\rm K}(x_1,x_2,x_3)=\begin{pmatrix} 
x_1 & x_2 & 0\\ 
1-x_1 & 1-x_2 & 0 \\
0 & 0 & x_3
\end{pmatrix},
\end{equation}
such that for $x_3\rightarrow 1$, ${\rm K}(x_1,x_2,x_3)\rightarrow {\rm M}(x_1,x_2)$. We consider the $6\times 6$ extensions of matrix ${\rm K}(x_1,x_2,x_3)$, namely the following
\begin{align*}
    & {\rm K}^6_{123}= \begin{pmatrix} 
x_1 & x_2 & 0 & 0 & 0 & 0\\ 
1-x_1 & 1-x_2 & 0 & 0 & 0 & 0\\
0 & 0 & x_3 & 0 & 0 & 0\\
0 & 0 & 0 & 1 & 0 & 0 \\
0 & 0 & 0 & 0 & 1 & 0 \\
0 & 0 & 0 & 0 & 0 & 1 \\
\end{pmatrix}, \quad
{\rm K}^6_{145}= \begin{pmatrix} 
x_1 & 0 & 0 & x_2 & 0 & 0\\ 
0 & 1 & 0 & 0 & 0 & 0 \\
0 & 0 & 1 & 0 & 0 & 0 \\
1-x_1 & 0 & 0 & 1-x_2 & 0 & 0\\
0 & 0 & 0 & 0 & x_3 & 0\\
0 & 0 & 0 & 0 & 0 & 1 \\
\end{pmatrix},\\
& {\rm K}^6_{246}= \begin{pmatrix} 
1 & 0 & 0 & 0 & 0 & 0 \\
0 & x_1 & 0 & x_2 & 0 & 0\\ 
0 & 0 & 1 & 0 & 0 & 0 \\
0 & 1-x_1 & 0 & 1-x_2 & 0 & 0\\
0 & 0 & 0 & 0 & 1 & 0 \\
0 & 0 & 0 & 0 & 0 & x_3
\end{pmatrix},\quad 
 {\rm K}^6_{356}= \begin{pmatrix} 
 1 & 0 & 0 & 0 & 0 & 0 \\
 0 & 1 & 0 & 0 & 0 & 0 \\
0 & 0 & x_1 & 0 & x_2 & 0\\ 
0 & 0 & 0 & 1 & 0 & 0 \\
0 & 0 & 1-x_1 & 0 & 1-x_2 & 0\\
0 & 0 & 0 & 0 & 0 & x_3
\end{pmatrix},
\end{align*}
and substitute them to the local tetrahedron equation
\begin{align*}
  &{\rm K}^6_{123}(u_1,u_2,u_3){\rm K}^6_{145}(v_1,v_2,v_3){\rm K}^6_{246}(w_1,w_2,w_3){\rm K}^6_{356}(r_1,r_2,r_3)=\\
  &{\rm K}^6_{356}(t_1,t_2,t_3){\rm K}^6_{246}(z_1,z_2,z_3){\rm K}^6_{145}(y_1,y_2,y_3){\rm K}^6_{123}(x_1,x_2,x_3).
\end{align*}
The above implies the following correspondence:
\begin{subequations}\label{corr-KP}
\begin{align}
& u_1=\frac{x_1y_1}{z_1+x_1(y_1-z_1+z_2-y_1z_2)},\\
 & u_2=\frac{y_2(1-z_1)+x_2y_1(z_2-1)+x_2y_2(z_1-z_2)}{(1-x_2)(z_2-z_1)+(1-z_2)(y_2-x_2y_1)},\\
 &u_3=\frac{(t_1-t_2)x_3y_3}{(t_1-1)x_3-(t_2-1)y_3},\\
& v_1= z_1+x_1(y_1-z_1+z_2-y_1z_2),\\
& v_2=\frac{\left[z_1+x_1(y_1-z_1+z_2-y_1z_2)\right] \left[(1-x_2)y_2z_1-(y_1-y_2)x_2z_2\right]}{x_1 (y_1-y_2) (z_1-z_2)-z_1(x_2y_1-y_2)},\\
&v_3=\frac{(t_1-t_2)x_3y_3}{t_1x_3-t_2y_3},\\
&w_1=\frac{(x_1-x_2)y_1z_1\left[y_1x_2(z_2-1)+(z_1-z_2)(x_2-1)-y_2(z_2-1)\right]}{\Lambda},\\
& w_2=\frac{\left[x_1z_2(y_1-y_2)+y_2z_1(x_1-1)\right]\left[y_1x_2(z_2-1)+(z_1-z_2)(x_2-1)-y_2(z_2-1)\right]}{\Lambda},\\
&r_1=\frac{(t_1-1) x_3-(t_2-1)y_3}{(t_1-t_2)y_3}t_1,\\
&r_2=\frac{(t_1-1)x_3-(t_2-1)y_3}{(t_1-t_2)x_3}t_2,\\
&r_3=\frac{t_3z_3}{w_3},
\end{align}
\end{subequations}
where $\Lambda$ is given by \eqref{lambda-expr}.

However, for the choice of the free variable $w_3=t_3$ the above correspondence defines a 4-simplex map. In particular, we have the following.

\begin{theorem}\label{KP-map}
The map $S\in\End(A^{12})$ given by
\begin{subequations}\label{4-simplex-KP}
\begin{align}
    x_1 &\mapsto u_1=\frac{x_1y_1}{z_1+x_1(y_1-z_1+z_2-y_1z_2)},\\
    x_2 &\mapsto u_2=\frac{y_2(1-z_1)+x_2y_1(z_2-1)+x_2y_2(z_1-z_2)}{(1-x_2)(z_2-z_1)+(1-z_2)(y_2-x_2y_1)},\\
    x_3 &\mapsto  u_3=\frac{(t_1-t_2)x_3y_3}{(t_1-1)x_3-(t_2-1)y_3},\\
    y_1 &\mapsto v_1= z_1+x_1(y_1-z_1+z_2-y_1z_2),\\
    y_2 &\mapsto v_2=\frac{\left[z_1+x_1(y_1-z_1+z_2-y_1z_2)\right] \left[(1-x_2)y_2z_1-(y_1-y_2)x_2z_2\right]}{x_1 (y_1-y_2) (z_1-z_2)-z_1(x_2y_1-y_2)},\\
    y_3 &\mapsto v_3=\frac{(t_1-t_2)x_3y_3}{t_1x_3-t_2y_3},\\
    z_1 &\mapsto w_1=\frac{(x_1-x_2)y_1z_1\left[y_1x_2(z_2-1)+(z_1-z_2)(x_2-1)-y_2(z_2-1)\right]}{\Lambda},\\
    z_2 &\mapsto w_2=\frac{\left[x_1z_2(y_1-y_2)+y_2z_1(x_1-1)\right]\left[y_1x_2(z_2-1)+(z_1-z_2)(x_2-1)-y_2(z_2-1)\right]}{\Lambda},\\
    z_3 &\mapsto w_3=t_3,\\
    t_1 &\mapsto r_1=\frac{(t_1-1) x_3-(t_2-1)y_3}{(t_1-t_2)y_3}t_1,\\
    t_2 &\mapsto r_2=\frac{(t_1-1)x_3-(t_2-1)y_3}{(t_1-t_2)x_3}t_2,\\
    t_3 &\mapsto r_3=z_3,
\end{align}
\end{subequations}
where $\Lambda$ is given by \eqref{lambda-expr}, is a twelve-dimensional, noninvolutive 4-simplex map which possesses the following functionally independent invariants 
\begin{equation}\label{KP-invariants}
    I_1=x_1y_1,\quad I_2=(y_2-1)(z_2-1),\quad I_3=x_3t_1,\quad I_4=y_3 (1-t_2),\quad I_5=z_3t_3,\quad I_6=z_3+t_3.
\end{equation}

Moreover, map $S$ is birational.
\end{theorem}
\begin{proof}
Map \eqref{4-simplex-KP} follows after substitution of $w_3=t_3$ to \eqref{corr-KP}. It can be verified that \eqref{4-simplex-KP} is a 4-simplex map by straightforward substitution to the 4-simplex equation \eqref{4-simplex-eq}. Now, since $v_3\circ S =\frac{(t_1-t_2)x_3y_3(t_1x_3-t_2y_3)}{t_1x_3(t_1x_3-2t_2y_3)+t_2y_3\left[x_3+(t_2-1)y_3\right]} \neq y_3$, it follows that $S\circ S\neq\id$, therefore map \eqref{4-simplex-KP} is noninvolutive.

Now, we have that $u_1v_1\stackrel{\eqref{4-simplex-KP}}{=}x_1y_1$, $(v_2-1)(w_2-1)\stackrel{\eqref{4-simplex-KP}}{=}(y_2-1)(z_2-1)$, $u_3r_1\stackrel{\eqref{4-simplex-KP}}{=}x_3t_1$, $v_3 (1-r_2)\stackrel{\eqref{4-simplex-KP}}{=}y_3 (1-t_2)$, $w_3r_3\stackrel{\eqref{4-simplex-KP}}{=}z_3t_3$ and $w_3+r_3\stackrel{\eqref{4-simplex-KP}}{=}z_3+t_3$, thus the quantities $I_i$, $i=1,\ldots 6$ in \eqref{KP-invariants} are invariants of map \eqref{4-simplex-KP}. Also, if $\rho_i=\nabla I_i$, $i=1,\ldots 6$, then the $6\times 12$ matrix $\left[\rho_1 \cdots \rho_6\right],$ consisting of vector columns $\rho_i$, $i=1,\ldots 6$ has rank 6, therefore $I_i$, $i=1,\ldots 6$ are functionally independent.

Finally, the inverse of map \eqref{4-simplex-KP} is given by:
\begin{subequations}\label{inverse-4-simplex-KP}
\begin{align}
    u_1 &\mapsto x_1=\frac{u_1v_1\left[w_1 (u_2-1)(v_2-1)+(v_2-v_1)(1-u_1+(u_1-u_2)w_2)\right]}{E},\\
    u_2 &\mapsto x_2=\frac{\left[u_1 v_2 (w_1-1)-u_2w_1\right] \left[(v_2-1)(u_2-1)w_1+(v_1-v_2) (u_1-1-(u_1-u_2)w_2)\right]}{-E},\\
    u_3 &\mapsto  x_3=r_1 (u_3-v_3)+v_3,\\
    v_1 &\mapsto y_1=\frac{E}{w_1 (u_2-1)(v_2-1)+(v_2-v_1)(1-u_1+(u_1-u_2)w_2)},\\
    v_2 &\mapsto y_2=u_2w_2-u_1v_2(w_2-1),\\
    v_3 &\mapsto y_3=r_2 (u_3-v_3)+v_3,\\
    w_1 &\mapsto z_1=\frac{(u_1-u_2)v_1w_1}{u_1(v_1-v_2)+(u_1v_2-u_2)w_1},\\
    w_2 &\mapsto z_2=\frac{v_2(u_1-1)(w_2-1)-w_2(u_2-1)}{1-u_1v_2+(u_1v_2-u_2)w_2},\\
    w_3 &\mapsto z_3=t_3,\\
    r_1 &\mapsto t_1=\frac{u_3r_1}{u_3r_1+v_3(1-r_1)}\\
    r_2 &\mapsto t_2=\frac{u_3r_2}{u_3r_2+v_3(1-r_2)},\\
    r_3 &\mapsto t_3=z_3.
\end{align}
\end{subequations}
where 
\begin{align*}
    E=&u_1^2v_2(v_1-v_2)(w_1-1)(w_2-1)+u_2w_1\left[v_1(u_2w_2-1)+v_2-w_2(u_2+v_2-1)\right]+\\
    &u_1v_1\left[w_1(1+u_2v_2)-v_2\right]+u_1v_1w_2\left[u_2-1+v_2-u_2w_1(v_2+1)\right]+\\
    &u_1v_2\left[v_2(w_1-1)(w_2-1)+w_2-u_2(w_1+w_2)+(2u_2-1)w_1w_2\right]
\end{align*}
Therefore, map \eqref{4-simplex-KP} is birational.
\end{proof}

\section{Darboux transformations scheme}\label{Darboux_scheme}
In this section, we demonstrate a method for constructing 4-simplex maps, which can be restricted to parametric 4-simplex maps on the level sets of the determinant of the associated Lax matrix, using Darboux transformations. As an illustrative example of this method we use Darboux transformations related to the NLS equation, and we construct NLS type of birational parametric 4-simplex maps.

\subsection{Darboux transformations}
Let ${\bm{u}}_t=F({\bm{u}},{\bm{u}}_x,{\bm{u}}_{xx},\ldots)$, ${\bm{u}}=(u_1(x,t),u_2(x,t),\ldots u_N(x,t))$, be a system of evolution type integrable PDEs with Lax pair $(\mathcal{L},\mathcal{A})$, where $\mathcal{L}={\rm{D}}_x-{\rm U}(u;\lambda)$ and $\mathcal{A}={\rm{D}}_t-{\rm V}(u;\lambda)$. That is,
$$
{\bm{u}}_t=F({\bm{u}},{\bm{u}}_x,{\bm{u}}_{xx},\ldots)~~ \Leftrightarrow ~~ {\rm{U}}_t-{\rm{V}}_x+{\rm{U}}{\rm{V}}-{\rm{V}}{\rm{U}}=0.
$$

\begin{definition}\label{def-Darboux}
A Darboux transformation is an invertible matrix $\rm{B}$, such that
\begin{gather}
\label{Def-Darboux}
\rm{B}\big(D_x-\rm{U}(\bm{u};\lambda)\big)\rm{B}^{-1}=
D_x-U({\bm{u}_{10}};\lambda),\qquad \rm{B}\big(D_t-\rm{V}(\bm{u};\lambda)\big)\rm{B}^{-1}=
D_t-\rm{V}({\bm{u}_{10}};\lambda).
\end{gather}
Matrix $\rm{B}$ is called Darboux matrix. 
\end{definition}

For example, consider the NLS equation, which in its most popular form appears as the coupled NLS system
\begin{equation}\label{NLS-eq}
    p_t=\frac{1}{2}p_{xx}-4p^2q,\quad  q_t=-\frac{1}{2}q_{xx}+4pq^2,
\end{equation}
where $p=p(x,t)$ and $q=q(x,t)$, $x\in\mathbb{R}$, $t>0$, and indices denote partial derivatives. System \eqref{NLS-eq} possesses the following Lax pair
\begin{equation}\label{Lax-NLS}
\mathcal{L} =D_x-\lambda \left(\begin{array}{cc}1 & 0\\0 & -1\end{array}\right)-\left(\begin{array}{cc}0 & 2p\\2q & 0\end{array}\right),\quad
\mathcal{A}=D_t-\lambda^2 \left(\begin{array}{cc}1 & 0\\0 & -1\end{array}\right)-\lambda\left(\begin{array}{cc}0 & 2p\\2q & 0\end{array}\right)
\end{equation} 

A Darboux transformation for $\mathcal{L}$ and $\mathcal{A}$ in \eqref{Lax-NLS} is:
\begin{equation}\label{DM-NLS}
  {\rm B}=\lambda \left(
     \begin{array}{cc}
         1 & 0\\
         0 & 0
     \end{array}\right)+\left(
     \begin{array}{cc}
         f & p\\
         \tilde{q} & 1
     \end{array}\right),
\end{equation}
where its entries obey the system of equations
$\partial_x f=2 (pq-\tilde{p}\tilde{q}), \quad \partial_x p =2 (p f -\tilde{p}),\quad \partial_x \tilde{q}=2 ( q-\tilde{q}f)$, i.e. the so-called B\"acklund transformation. A first integral of this  system of differential equations is 
$\partial_x (f-p\tilde{q})=0$, which implies that $\partial_x (\det{\rm B})=\partial_x (f-p\tilde{q})=0$, i.e. $\det{\rm B}=const.$

\subsection{Formulation of the method}
This is a generalisation of the methods introduced in \cite{Sokor-Sasha, Sokor-2020} for constructing parametric Yang--Baxter and tetrahedron maps. The basic steps of the method are summarised in Figure \ref{Darboux scheme}. 

In particular, the method consists of three basic steps:

\textbf{Step $I$}: We start with a Darboux matrix, ${\rm B(x;k)}$, which generates a tetrahedron map via the local Yang--Baxter equation. We consider its $3\times 3$ extension ${\rm M(x;k)}={\rm B^3_{12}(x;k)}$, and let ${\rm K}(x_1,x_2;\kappa)$ be a generalisation of ${\rm M(x;k)}$, as in \eqref{K-matrix}, such as ${\rm K}(x_1,x_2;\kappa)\rightarrow {\rm M(x_1;k)}$, for $x_2\rightarrow 1$. Then, we substitute ${\rm K}(x_1,x_2;\kappa)$ into the local tetrahedron equation \eqref{local-tetrahedron}, and we derive a 4-simplex map.

\textbf{Step $II$}: Matrix ${\rm M(x;k)}$, being a $3\times 3$ generalisation of a Darboux matrix, has constant determinant. We demand that its generalisation, ${\rm K}(x_1,x_2;\kappa)$ has also constant determinant. Then, we restrict the 4-simplex map to lower-dimensional parametric maps on the level sets $\det{\rm K}(x_1,x_2;\kappa)=C$.

\textbf{Step $III$}: These parametric maps are solutions to the parametric 4-simplex equation \eqref{Par-4-simplex-eq}.

\begin{figure}[ht]
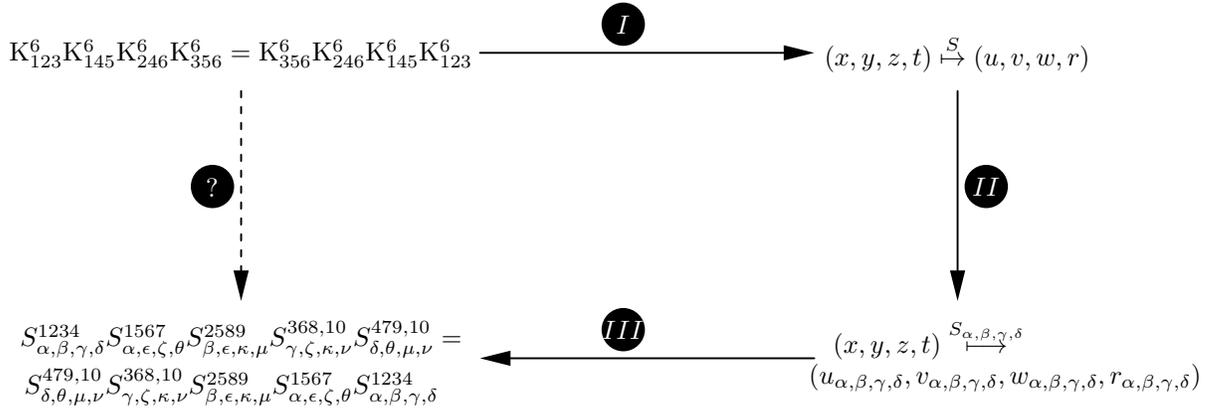

\centering
\centertexdraw{ 
\setunitscale 0.5
\move (-2.5 2.2)  \arrowheadtype t:F \avec(1 2.2)
\move (-5 1.8) \lpatt(0.067 0.1) \arrowheadtype t:F \avec(-5 -.4) \lpatt()
\move (2.5 1.8)  \arrowheadtype t:F \avec(2.5 -.4) 
\move (1 -1)  \arrowheadtype t:F \avec(-2.5 -1) 
\textref h:C v:C \htext(-5 2.2){\small ${\rm K}^6_{123}{\rm K}^6_{145}{\rm K}^6_{246}{\rm K}^6_{356}=
  {\rm K}^6_{356}{\rm K}^6_{246}{\rm K}^6_{145}{\rm K}^6_{123}$}
\textref h:C v:C \htext(2.5 2.2){\small $(x,y,z,t)\stackrel{S}{\mapsto}(u,v,w,r)$}
\textref h:C v:C \htext(2.2 -.8){\small $(x,y,z,t)\stackrel{S_{\alpha,\beta,\gamma,\delta}}{\longmapsto}$}
\textref h:C v:C \htext(3 -1.2){\small $(u_{\alpha,\beta,\gamma,\delta},v_{\alpha,\beta,\gamma,\delta},w_{\alpha,\beta,\gamma,\delta},r_{\alpha,\beta,\gamma,\delta})$}
\textref h:C v:C \htext(-5 -.8){\small $S^{1234}_{\alpha,\beta,\gamma,\delta} S^{1567}_{\alpha,\epsilon,\zeta,\theta}  S^{2589}_{\beta,\epsilon,\kappa,\mu} S^{368,10}_{\gamma,\zeta, \kappa, \nu} S^{479,10}_{\delta,\theta,\mu,\nu}=$} 
\htext(-5.1 -1.3){\small $S^{479,10}_{\delta,\theta,\mu,\nu} S^{368,10}_{\gamma,\zeta, \kappa, \nu} S^{2589}_{\beta,\epsilon,\kappa,\mu} S^{1567}_{\alpha,\epsilon,\zeta,\theta}  S^{1234}_{\alpha,\beta,\gamma,\delta}$}
  \move (-5.3 .8)\fcir f:0 r:.23
\textref h:C v:C \small{\color{white} \htext(-5.3 .8){?}}
\move (2.8 .8)\fcir f:0 r:.23
\textref h:C v:C \small{\color{white} \htext(2.8 .8){$II$}}
\move (-1 2.5)\fcir f:0 r:.23
\textref h:C v:C \small{\color{white}\htext(-1 2.5){$I$}}
\move (-1 -.7)\fcir f:0 r:.23
\textref h:C v:C \small{\color{white}\htext(-1 -.7){$III$}}
}
\caption{Darboux construction scheme.}\label{Darboux scheme}
\end{figure}

In the next sections, we apply this method to the NLS type Darboux matrix \eqref{DM-NLS} in order to construct 4-simplex maps.

\subsection{NLS type 4-simplex maps}
Changing $(p,\tilde{q},f+\lambda)\rightarrow (x_1,x_2,X)$ in \eqref{DM-NLS}, we define the following matrix
\begin{equation}\label{M-NLS}
  {\rm B}(x_1,x_2,X)=\left(
     \begin{array}{cc}
         X & x_1\\
         x_2 & 1
     \end{array}\right).
\end{equation}
This matrix was used in \cite{Sokor-2020} to derive a birational, noninvolutive tetrahedron map by substitution of \eqref{M-NLS} to the local Yang--Baxter equation.

Now, we consider its $3\times 3$ extension of ${\rm B}$, namely ${\rm M}(x_1,x_2,X)\equiv {\rm B}^3_{12}(x_1,x_2,X)=\begin{pmatrix} 
X & x_1 & 0\\ 
x_2 & 1 & 0\\
0 & 0 & 1
\end{pmatrix}$, and let ${\rm K}(x_1,x_2,x_3,X)$, given by
\begin{equation}\label{NLS-Lax}
    {\rm K}(x_1,x_2,x_3,X)=\begin{pmatrix} 
X & x_1 & 0\\ 
x_2 & 1 & 0\\
0 & 0 & x_3
\end{pmatrix},
\end{equation}
such that ${\rm K}(x_1,x_2,x_3,X)\rightarrow {\rm M}(x_1,x_2,X)$, for $x_3\rightarrow 1$. Next, we substitute ${\rm K}(x_1,x_2,x_3,X)$ to the local tetrahedron equation
\begin{align}
  &{\rm K}^6_{123}(u_1,u_2,u_3,U){\rm K}^6_{145}(v_1,v_2,v_3,V){\rm K}^6_{246}(w_1,w_2,w_3,W){\rm K}^6_{356}(r_1,r_2,r_3,R)=\nonumber\\
  &{\rm K}^6_{356}(t_1,t_2,t_3,T){\rm K}^6_{246}(z_1,z_2,z_3,Z){\rm K}^6_{145}(y_1,y_2,y_3,Y){\rm K}^6_{123}(x_1,x_2,x_3,X),\label{NLS-local-tetra}
\end{align}
which is equivalent with the following correspondence:
\begin{subequations}\label{NLS-correspondence}
\begin{align}
    u_1&=\frac{x_1(y_1y_2-Y)+y_1z_2}{z_1z_2-Z},\quad u_2=\frac{x_2Z+y_2z_1X}{XY}U,\label{NLS-correspondence-a}\\
    v_1&=\frac{y_1Z+x_1z_1(y_1y_2-Y)}{y_1y_2z_1(x_1y_2+z_2)X-(x_1y_2z_1+z_1z_2-Z)XY+x_2[y_1z_2+x_1(y_1y_2-Y)]Z}\frac{XY}{U}\label{NLS-correspondence-b}\\
    v_2&=x_2z_2+y_2X,\quad v_3=y_3,\label{NLS-correspondence-c}\\
    V&=\frac{XY}{U},\label{NLS-correspondence-d}\\
    w_1&=\frac{[x_2y_1Z+(y_1y_2-Y)z_1X](z_1z_2-Z)}{y_1y_2z_1(x_1y_2+z_2)X-(x_1y_2z_1+z_1z_2-Z)XY+x_2[y_1z_2+x_1(y_1y_2-Y)]Z}\label{NLS-correspondence-e}\\
    w_2&=x_1y_2+z_2,\label{NLS-correspondence-f}\\ W&=\frac{(x_1x_2-X)(z_1z_2-Z)YZ}{y_1y_2z_1(x_1y_2+z_2)X-(x_1y_2z_1+z_1z_2-Z)XY+x_2[y_1z_2+x_1(y_1y_2-Y)]Z},\label{NLS-correspondence-g}\\
    r_1&=\frac{t_1y_3}{u_3},\quad r_2=\frac{t_2x_3}{y_3},\quad r_3=\frac{t_3z_3}{w_3},\quad R=\frac{Tx_3}{u_3}.\label{NLS-correspondence-h}
\end{align}
\end{subequations}
The above correspondence does not satisfy the 4-simplex equation for arbitrary choice of $u_3$, $U$ and $w_3$. However, as we shall see below, there is at least a choice of $U$ for which \eqref{NLS-correspondence} defines a 4-simplex map.

In particular, the determinant of equation \eqref{NLS-local-tetra} is
$$
(U-u_1u_2)u_3(V-v_1v_2)v_3(W-w_1w_2)w_3(R-r_1r_2)r_3=(X_1-x_1x_2)x_3(Y-y_1y_2)y_3(Z-z_1z_2)z_3(T-t_1t_2)t_3.
$$
We choose $(U-u_1u_2)u_3=(X_1-x_1x_2)x_3$, $(W-w_1w_2)w_3=(Z-z_1z_2)z_3$ and $u_3=x_3$. Then, the following holds.
\begin{proposition}
The system consisting of equation \eqref{NLS-local-tetra} together with equations 
$$
(U-u_1u_2)u_3=(X-x_1x_2)x_3,\quad (W-w_1w_2)w_3=(Z-z_1z_2)z_3,\quad u_3=x_3,
$$
has a unique solution, namely a map 
$
(x_1,x_2,x_3,X,y_1,y_2,y_3,Y,z_1,z_2,z_3,Z,t_1,t_2,t_3,T)\overset{S}{\longrightarrow} (u_1,u_2,u_3,U,v_1,v_2,v_3,V,w_1,w_2,w_3,W,r_1,r_2,r_3,R),
$
given by,
\begin{subequations}\label{4-simplex-NLS-16D} 
\begin{align}
x_1\mapsto u_1 &=\frac{x_1(y_1y_2-Y)+y_1z_2}{z_1z_2-Z},\label{4-simplex-NLS-16D-a}\\
x_2\mapsto u_2 &=\frac{(x_1x_2-X)(y_2z_1X+x_2Z)(z_1z_2-Z)}{y_1y_2z_1(x_1y_2+z_2)X-(x_1y_2z_1+z_1z_2-Z)XY+x_2[y_1z_2+x_1(y_1y_2-Y)]Z},\label{4-simplex-NLS-16D-b}\\
x_3\mapsto u_3 &=x_3,\label{4-simplex-NLS-16D-c}\\
X\mapsto U &=\frac{(x_1x_2-X)(z_1z_2-Z)X}{y_1y_2z_1(x_1y_2+z_2)X-(x_1y_2z_1+z_1z_2-Z)XY+x_2[y_1z_2+x_1(y_1y_2-Y)]Z},\label{4-simplex-NLS-16D-d}\\
y_1\mapsto v_1 &=\frac{x_1z_1(y_1y_2-Y)+y_1Z}{(x_1x_2-X)(z_1z_2-Z)},\label{4-simplex-NLS-16D-e}\\
y_2\mapsto v_2 &=x_2z_2+y_2X,\label{4-simplex-NLS-16D-f}\\
y_3\mapsto v_3 &=y_3,\label{4-simplex-NLS-16D-g}\\
Y\mapsto V &=\frac{y_1y_2z_1(x_1y_2+z_2)X-(x_1y_2z_1+z_1z_2-Z)XY+x_2[y_1z_2+x_1(y_1y_2-Y)]Z}{(x_1x_2-X)(z_1z_2-Z)},\label{4-simplex-NLS-16D-h}\\
z_1\mapsto w_1 &=\frac{[x_2y_1Z+z_1(y_1y_2-Y)X](z_1z_2-Z)}{y_1y_2z_1(x_1y_2+z_2)X-(x_1y_2z_1+z_1z_2-Z)XY+x_2[y_1z_2+x_1(y_1y_2-Y)]Z},\label{4-simplex-NLS-16D-i}\\
z_2\mapsto w_2 &=x_1y_2+z_2,\label{4-simplex-NLS-16D-j}\\
z_3\mapsto w_3 &=z_3,\label{4-simplex-NLS-16D-k}\\
Z\mapsto W &=\frac{(x_1x_2-X)(z_1z_2-Z)YZ}{y_1y_2z_1(x_1y_2+z_2)X-(x_1y_2z_1+z_1z_2-Z)XY+x_2[y_1z_2+x_1(y_1y_2-Y)]Z},\label{4-simplex-NLS-16D-l}\\
t_1\mapsto r_1 &=\frac{t_1y_3}{x_3},\label{4-simplex-NLS-16D-m}\\
t_2\mapsto r_2 &=\frac{t_2x_3}{y_3},\label{4-simplex-NLS-16D-n}\\
t_3\mapsto t_3 &=r_3,\label{4-simplex-NLS-16D-0}\\
T\mapsto R &=T.\label{4-simplex-NLS-16D-p}
\end{align}
\end{subequations}
Map \eqref{4-simplex-NLS-16D} is a sixteen-dimensional noninvolutive 4-simplex map.
\end{proposition}
\begin{proof}
The system consisting of equations \eqref{NLS-correspondence-a}, \eqref{NLS-correspondence-c}, \eqref{NLS-correspondence-e}, \eqref{NLS-correspondence-f}, \eqref{NLS-correspondence-g}  and equations $(U-u_1u_2)u_3=(X_1-x_1x_2)x_3$, $(V-v_1v_2)v_3=(Y-y_1y_2)y_3,$ $(W-w_1w_2)w_3=(Z-z_1z_2)z_3,$  $u_3=x_3,$ has a unique solution given by \eqref{4-simplex-NLS-16D-a}--\eqref{4-simplex-NLS-16D-d}, \eqref{4-simplex-NLS-16D-f}, \eqref{4-simplex-NLS-16D-g} and \eqref{4-simplex-NLS-16D-i}--\eqref{4-simplex-NLS-16D-k}. Moreover, substituting $U$ given by \eqref{4-simplex-NLS-16D-d} to \eqref{NLS-correspondence-d} and $u_3=x_3$, $w_3=z_3$ to  \eqref{NLS-correspondence-d} and \eqref{NLS-correspondence-h}, respectively, we obtain $V$ and $r_1, r_2, r_3$ and $R$ given in \eqref{4-simplex-NLS-16D-h} and \eqref{4-simplex-NLS-16D-m}--\eqref{4-simplex-NLS-16D-p}, respectively.
The 4-simplex property can be readily verified by substitution to the 4-simplex equation. Finally, for the involutivity of the map we have 
$$
r_1(u_1,u_2,u_3,U,v_1,v_2,v_3,V,w_1,w_2,w_3,W,r_1,r_2,r_3,R)=\frac{t_1y_3^3}{x_3^2}
$$
That is, $T\circ T\neq\id$. Thus, map \eqref{4-simplex-NLS-16D} is noninvolutive.
\end{proof}

\subsection{NLS type parametric 4-simplex maps: Restriction on the level sets of the invariants}
Here, we restrict map \eqref{4-simplex-NLS-16D} to a twelve-dimensional noninvolutive parametric 4-simplex map. In particular, we have the following.

\begin{theorem}
\begin{enumerate}
    \item The quantities $\Xi=(X-x_1x_2)x_3$, $\Psi=(Y-y_1y_2)y_3$, $\Phi=(Z-z_1z_2)z_3$ and $\Omega=(T-t_1t_2)t_3$ are invariants of map \eqref{4-simplex-NLS-16D}.
    \item Map \eqref{4-simplex-NLS-16D} can be restricted to a twelve-dimensional birational parametric 4-simplex map
    $$
(x_1,x_2,x_3,y_1,y_2,y_3,z_1,z_2,z_3,t_1,t_2,t_3)\overset{S_{\alpha,\beta,\gamma,\delta}}{\longrightarrow}(u_1,u_2,u_3,v_1,v_2,v_3,w_1,w_2,w_3,r_1,r_2,r_3),
$$
given by:
\begin{subequations}\label{4-simplex-NLS-12D} 
\begin{align}
x_1\mapsto u_1 &=\frac{(\beta x_1-y_1y_3z_2)z_3}{cy_3},\\
x_2\mapsto u_2 &=\frac{\alpha\gamma \left[ \gamma x_2x_3 +(z_2+x_1y_2)x_2x_3z_1z_3+\alpha y_2 z_1z_3\right]}{\alpha\beta\gamma-\left[x_2x_3 (x_1y_2+z_2)+\alpha y_2\right]\left[z_1z_3 (\beta x_1-y_1y_3z_2)-\gamma y_1y_3\right]}\cdot\frac{y_3}{x_3z_3},\\
x_3\mapsto u_3 &=x_3,\\
y_1\mapsto v_1 &=\frac{x_3\left[\gamma y_1y_3+(y_1y_3z_2-\beta x_1)z_1z_3\right]}{\alpha\gamma y_3},\\
y_2\mapsto v_2 &=(x_1y_2+z_2)x_2+\alpha\frac{y_2}{x_3},\\
y_3\mapsto v_3 &=y_3,\\
z_1\mapsto w_1 &=\frac{x_2x_3y_1y_3(\gamma+z_1z_2z_3)-\beta z_1z_3(\alpha+x_1x_2x_3)}{\left[x_2x_3 (x_1y_2+z_2)+\alpha y_2\right]\left[z_1z_3 (\beta x_1-y_1y_3z_2)-\gamma y_1y_3\right]-\alpha\beta\gamma}\cdot\frac{\gamma}{z_3},\\
z_2\mapsto w_2 &=x_1y_2+z_2,\\
z_3\mapsto w_3 &=z_3,\\
t_1\mapsto r_1 &=\frac{t_1y_3}{x_3},\\
t_2\mapsto r_2 &=\frac{t_2x_3}{y_3},\\
t_3\mapsto r_3 &=t_3,
\end{align}
\end{subequations}
on the invariant leaves
    \begin{align}
    A_{\alpha}&:=\{(x_1,x_2,x_3,X)\in\mathbb{C}^4:X=\frac{\alpha+x_1x_2x_3}{x_3}\},\quad B_{\beta}:=\{(y_1,y_2,y_3,Y)\in\mathbb{C}^4:Y=\frac{\beta+y_1y_2y_3}{y_3}\},\nonumber\\ C_{\gamma}&:=\{(z_1,z_2,z_3,Z)\in\mathbb{C}^4:Z=\frac{\gamma+z_1z_2z_3}{z_3}\},\quad D_{\delta}:=\{(t_1,t_2,t_3,T)\in\mathbb{C}^4:T=\frac{\delta+t_1t_2t_3}{t_3}\}.\label{inv-leaves}
\end{align}
\item Map \eqref{4-simplex-NLS-12D} admits the following functional independent invariants
\begin{align}
   & I_1=\left(x_1x_2+\frac{\alpha}{x_3}\right)\left(y_1y_2+\frac{\beta}{y_3}\right),\quad I_2=\left(y_1y_2+\frac{\beta}{y_3}\right)\left(z_1z_2+\frac{\gamma}{x_3}\right),\quad I_3=x_3,\nonumber\\ 
  &I_4=y_3,\quad I_5=z_3,\quad I_6=t_1t_2.
\end{align}
\end{enumerate}
\end{theorem}

\begin{proof}
Regarding \textbf{1}. Indeed, $\Xi=(X-x_1x_2)x_3=(U-u_1u_2)u_3$, $\Psi=(Y-y_1y_2)y_3=(V-v_1v_2)v_3$, $\Phi=(Z-z_1z_2)z_3=(W-w_1w_2)w_3$ and $\Omega=(T-t_1t_2)t_3=(R-r_1r_2)r_3$ in view of \eqref{4-simplex-NLS-16D}.

With regards to \textbf{2}, we set $\Xi=\alpha$, $\Psi=\beta$, $\Phi=\gamma$ and $\Omega=\delta$. Now, using the conditions $X=\frac{\alpha+x_1x_2x_3}{x_3}$, $Y=\frac{\beta+y_1y_2y_3}{y_3}$, $Z=\frac{\gamma+z_1z_2z_3}{z_3}$ and $T=\frac{\delta+t_1t_2t_3}{t_3}$, we eliminate $X$, $Y$, $Z$ and $T$ from  the sixteen-dimensional map \eqref{4-simplex-NLS-16D}, and we obtain \eqref{4-simplex-NLS-12D}. It can be verified by substitution that map \eqref{4-simplex-NLS-16D} satisfies the parametric 4-simplex equation. For the involutivity, we have that
$$
r_1(u_1,u_2,u_3,v_1,v_2,v_3,w_1,w_2,w_3,t_1,t_2,t_3)=\frac{t_1y_3^2}{x_3^2}.
$$
Thus, $S_{\alpha,\beta,\gamma,\delta}\circ S_{\alpha,\beta,\gamma,\delta}\neq \id$, and the map is noninvolutive.

Finally, concerning \textbf{3}, it can be readily proven that $\left(u_1u_2+\frac{\alpha}{u_3}\right)\left(v_1v_2+\frac{\beta}{v_3}\right)\overset{\eqref{4-simplex-NLS-12D}}{=}\left(x_1x_2+\frac{\alpha}{x_3}\right)\left(y_1y_2+\frac{\beta}{y_3}\right)$, $\left(v_1v_2+\frac{\beta}{v_3}\right)\left(w_1w_2+\frac{\gamma}{w_3}\right)\overset{\eqref{4-simplex-NLS-12D}}{=}\left(y_1y_2+\frac{\beta}{y_3}\right)\left(z_1z_2+\frac{\gamma}{z_3}\right)$, $v_3\overset{\eqref{4-simplex-NLS-12D}}{=}y_3$, $w_3\overset{\eqref{4-simplex-NLS-12D}}{=}z_3$ and $r_1r_2\overset{\eqref{4-simplex-NLS-12D}}{=}t_1t_2$. Moreover,  if $\rho_i=\nabla I_i$, $i=1,\ldots 6$, then the $6\times 12$ matrix $\left[\rho_1 \cdots \rho_6\right],$ consisting of vector columns $\rho_i$, $i=1,\ldots 6$ has rank 6, therefore $I_i$, $i=1,\ldots 6$ are functionally independent.
\end{proof}

\begin{corollary}
The 4-simplex map has Lax representation
\begin{align*}
  &{\rm K}^6_{123}(u_1,u_2,u_3;\alpha){\rm K}^6_{145}(v_1,v_2,v_3;\beta){\rm K}^6_{246}(w_1,w_2,w_3,\gamma){\rm K}^6_{356}(r_1,r_2,r_3;\delta)=\\
  &{\rm K}^6_{356}(t_1,t_2,t_3,\delta){\rm K}^6_{246}(z_1,z_2,z_3,\gamma){\rm K}^6_{145}(y_1,y_2,y_3,\beta){\rm K}^6_{123}(x_1,x_2,x_3,\alpha).
\end{align*}
where
\begin{equation*}
    {\rm K}(x_1,x_2,x_3;\alpha)=\begin{pmatrix} 
x_1x_2+\frac{\alpha}{x_3} & x_1 & 0\\ 
x_2 & 1 & 0\\
0 & 0 & x_3
\end{pmatrix}.
\end{equation*}
\end{corollary}

\begin{remark}\normalfont
    The birational map \eqref{4-simplex-NLS-12D} has a property similar to  quadrirationality as defined in \cite{ABS}. 
   
   Specifically, let bold letters ${\bm x}$ denote vectors ${\bm x}=(x_1,x_2,x_3)$, where $x_i\in\mathcal{X}$, $i=1,2,3$. For any fixed $\bm{y},\bm{z},\bm{t}\in\mathcal{X}^3$, function ${\bm u}(\cdot,\bm{y},\bm{z},\bm{t}):\mathcal{X}^3\rightarrow \mathcal{X}^3$ is birational, for any fixed $\bm{x},\bm{z},\bm{t}\in\mathcal{X}^3$, function ${\bm v}(\bm{x},\cdot,\bm{z},\bm{t}):\mathcal{X}^3\rightarrow \mathcal{X}^3$ is birational, for any fixed $\bm{x},\bm{y},\bm{t}\in\mathcal{X}^3$, function ${\bm w}(\bm{x},\bm{y},\cdot,\bm{t}):\mathcal{X}^3\rightarrow \mathcal{X}^3$ is birational, and, also, for any fixed $\bm{x},\bm{y},\bm{z}\in\mathcal{X}^3$, function ${\bm r}(\bm{x},\bm{y},\bm{z},\cdot):\mathcal{X}^3\rightarrow \mathcal{X}^3$ is birational. 
   
   Now, by consecutive bold capital letters $\bm{XY}$ we denote a pair of vectors $(\bm{x},\bm{y})$, i.e. $\bm{XY}=(\bm{x},\bm{y})\in\mathcal{X}^6$. It  can  be readily verified that, for any fixed $\bm{z},\bm{t}$, function $\bm{UV}(\cdot,\cdot,\bm{z},\bm{t}):\mathcal{X}^6\rightarrow \mathcal{X}^6$ is birational, for any fixed $\bm{y},\bm{t}$, function $\bm{UW}(\cdot,\bm{y},\cdot,\bm{t}):\mathcal{X}^6\rightarrow \mathcal{X}^6$ is birational, for any fixed $\bm{y},\bm{z}$, function $\bm{UR}(\cdot,\bm{y},\bm{z},\cdot):\mathcal{X}^6\rightarrow \mathcal{X}^6$ is birational, for any fixed $\bm{x},\bm{t}$, function $\bm{VW}(\bm{x},\cdot,\cdot,\bm{t}):\mathcal{X}^6\rightarrow \mathcal{X}^6$ is birational, for any fixed $\bm{x},\bm{z}$, function $\bm{VR}(\bm{x},\cdot,\bm{z},\cdot):\mathcal{X}^6\rightarrow \mathcal{X}^6$ is birational, and, also, for any fixed $\bm{x},\bm{y}$, function $\bm{WR}(\bm{x},\bm{y},\cdot,\cdot):\mathcal{X}^6\rightarrow \mathcal{X}^6$ is birational.

   Analogously, it can be shown  that, for any fixed value of $\bm{x}, \bm{y}, \bm{z}, \bm{t}\in\mathcal{X}^3$, the corresponding functions $\mathcal{X}^9\rightarrow \mathcal{X}^9$  are also birational.
    
    This property of the birational  map \eqref{4-simplex-NLS-12D} is a higher-dimensional analogue of the quadrirationality property defined in \cite{ABS}. Note that another version of higher order quadrirationality --- the so-called $2^n$-rational property --- was defined in \cite{Pavlos-Maciej}.
\end{remark}

\section{Conclusions}\label{conclusions}
In this paper, we propose a method for constructing solutions to the 4-simplex equation as (nontrivial) extensions of tetrahedron maps. The method is based on the consideration of a more general matrix than the one that generates tetrahedron maps via the local Yang--Baxter equation. The 4-simplex extensions are derived via the local tetrahedron equation. This method was employed to construct new 4-simplex maps, namely maps \eqref{4-simplex-a}, \eqref{4-simplex-a-2}, \eqref{4-simplex-b}, \eqref{4-simplex-b-2}, \eqref{4-simplex-c} and \eqref{4-simplex-d}. All these maps can be reduced to trivial extensions of the Kashaev--Korepanov--Sergeev tetrahedron maps at the limit $(x_2,y_2,z_2)\rightarrow (1,1,1)$. Furthermore, we constructed a new, 4-simplex extension of the KP map appeared in \cite{Dimakis}, namely map \eqref{4-simplex-KP}.

In the second part of this paper, we extended the ideas introduced in \cite{Sokor-Sasha} and \cite{Sokor-2020} about constructing Yang--Baxter and tetrahedron maps, respectively, to a Darboux scheme for constructing parametric 4-simplex maps. We employed this scheme to construct an NLS type sixteen-dimensional 4-simplex map, namely map \eqref{4-simplex-NLS-16D}, which can be restricted to a parametric twelve-dimensional 4-simplex map on invariant leaves, namely map \eqref{4-simplex-NLS-12D}. For this construction we used a Darboux transformation of the NLS equation. 

Our results can be extended in the following ways:
\begin{enumerate}
    \item Study the integrability of the maps presented in this paper. All 4-simplex maps constructed in this paper have enough functionally independent first integrals which indicates that these maps are integrable. However, their Liouville integrability is an open problem.
    
    \item Our methods can be employed to construct solutions to the $n$-simplex equation for $n\geq 5$. For example, one can consider $5\times 5$ generalisations (similar to those described in section \ref{method}) of all Lax matrices that were used to generate 4-simplex maps in this paper to construct 5-simplex maps.
    
    \item It is known that the Yang--Baxter equation and the tetrahedron equation are related to two- and three-dimensional integrable lattice equations. It is expected the 4-simplex equation is related to four-dimensional lattice equations. One could study this relation by extending known methods in relation to the Yang--Baxter equation and the tetrahedron equation. For instance, the symmetries of the associated lattice equations \cite{Pap-Tongas-Veselov, Kassotakis-Tetrahedron}, the existence of integrals in separable variables \cite{Pavlos-Maciej-2}, lifts and squeeze downs \cite{Kouloukas-Dihn, Pap-Tongas, Sokor-Kouloukas}.
    
    \item For the solutions constructed using Darboux transformations it is expected that the corresponding 4-simplex maps are certain B\"acklund type of transformations for the associated PDEs. 
    
    \item In this paper we used the local tetrahedron equation to construct novel solutions to the 4-simplex equation \eqref{local-tetrahedron}. Following \cite{Sergei-Sokor}, a matrix four-factorisation problem other than \eqref{local-tetrahedron} can be employed to construct new 4-simplex maps. Namely, one can replace $a(x,\kappa), b(x,\kappa), c(x,\kappa), d(x,\kappa), e(x,\kappa), f(x,\kappa), k(x,\kappa),  l(x,\kappa)$ and $m(x,\kappa)$  in \eqref{6x6-extensions} by matrices $n\times n$, and $1$ can be replaced by a matrix $m\times m$, $m\neq n$. Analogously to \cite{Sergei-Sokor}, we expect that this consideration will give rise to new, integrable 4-simplex maps. 
    
    \item Construct Grassmann extended versions of the 4-simplex maps constructed in this paper using Darboux transformations. In \cite{GKM, Sokor-2020-1}, Grassmann extended Darboux transformations for the NLS equation ere employed to construct Grassmann extended Yang--Baxter maps. One could consider a Grassmann extension of matrix \eqref{DM-NLS} and apply the method presented in Section \ref{Darboux scheme} to derive novel Grassmann extended 4-simplex maps.
    
    \item Construct fully noncommutative analogues of the solutions presented in this paper. Noncommutative $2$- and $3$-simplex maps have been recently quite popular (see, e.g., \cite{Doliwa-Kashaev, Kassotakis-Kouloukas, Sokor-2022} and the references therein).  For proving that a noncommutative map satisfies the 4-simplex equation, one may need to consider matrix refactorisation problems. We plan to study the relation of the 4-simplex equation with matrix ten-factorisation problems (as an extension of the results in \cite{Kouloukas-Papageorgiou, Sokor-2022}) in a future publication.

    \item Our  method (see Section \ref{method}) uses matrix \eqref{K-matrix} in  order to construct new 4-simplex  maps as nontrivial extensions of  tetrahedron maps. However, there are other nontrivial $3\times 3$ extensions of matrix \eqref{matrix-L} which can be employed for the construction of other 4-simplex maps. Moreover, the possible forms of $3\times 3$ matrices  \eqref{matrix-L-simplex} which can be used to derive new  4-simplex maps can be classified. Such results will appear in our future publications.
\end{enumerate}

\section*{Acknowledgements} The work on sections 1, 2 and 3 is funded by the Russian Science Foundation (Grant No. grant No. 21-71-30011), whereas the work on sections 4 and 5 is funded by the Ministry of Science and Higher Education (Agreement No. 075-02-2021-1397). I would like to thank V. Bardakov and D. Talalaev for visiting Yaroslavl and motivated me to initiate this work as well as for various useful discussions. Also, I am grateful to S. Igonin for many useful discussions.

\end{document}